\documentclass[11pt,letterpaper]{article}

\usepackage[utf8]{inputenc}
\usepackage[english]{babel}
\usepackage{csquotes}

\usepackage[margin=1in]{geometry} 

\usepackage{setspace}
\usepackage{graphicx}

\usepackage{amsmath, amssymb, amsfonts}
\usepackage{amsthm, mathrsfs, mathtools}

\usepackage{hyperref}

\usepackage[backend=biber,style=alphabetic, minbibnames=99,sorting=nyt,sortcites=true,maxbibnames=99]{biblatex}
\usepackage{titling}

\usepackage{authblk}

\usepackage{enumitem} 
\usepackage{caption} 
\usepackage{subcaption} 
\usepackage{booktabs} 
\usepackage{microtype} 
\usepackage{dsfont} 
\usepackage{faktor} 

\usepackage[baseline]{euflag}

\let\epsilon\varepsilon

\let\phi\varphi

\theoremstyle{plain}
\newtheorem{theorem}{Theorem}[section]
\newtheorem{lemma}[theorem]{Lemma}

\newtheorem{proposition}[theorem]{Proposition}

\newtheorem{remark}[theorem]{Remark}
\newtheorem*{theorem*}{Theorem}

\newcommand*{\Tr}[1]{\mathop{}\!\mathrm{Tr}{\left(#1\right)}}
\newcommand*{\tr}[2]{\mathop{}\!\mathrm{Tr}_{#1}{\left(#2\right)}}
\newcommand{\Span}[1]{\mathrm{Span}\left\lbrace #1 \right\rbrace}

\def\id{\mathds{1}}

\newcommand{\op}{\mathrm{op}}

\newcommand{\ket}[1]{\ensuremath{\left|#1\right\rangle}}

\newcommand{\ketbra}[2]{|#1\rangle\langle #2|  }
\newcommand{\sandwich}[3]{\langle #1|#2 |#3\rangle  }

\DeclareMathOperator{\diag}{diag}

\newcommand{\cH}{\mathcal{H}}

\newcommand{\cA}{\mathcal{A}}
\newcommand{\cB}{\mathcal{B}}
\newcommand{\cK}{\mathcal{K}}

\let\B\relax 
 
\newcommand{\A}{A_{a|x}}
\newcommand{\B}{B_{b|y}}

\newcommand\blfootnote[1]{%
  \begingroup
  \renewcommand\thefootnote{}%
  \footnotetext{#1}%
  \addtocounter{footnote}{-1}%
  \endgroup
}

\usepackage{xcolor}

\hypersetup{
  pdftitle    = {Quantitative Tsirelson's Theorems via Approximate Schur's Lemma and Probabilistic Stampfli's Theorems},
  pdfauthor   = {Xiangling Xu; Marc-Olivier Renou; Igor Klep},
  pdfsubject  = {Operator algebras; matrix theory; quantum information theory},
  pdfkeywords = {almost-commuting matrices, commutator estimates, approximate Schur's lemma, probabilistic Stampfli theorem, Weyl-Heisenberg (clock-and-shift) matrices, finite-dimensional C*-algebras, tensor-product factorization, quantitative Tsirelson's theorem, approximate Tsirelson's theorem},
  pdflang     = {en},
  pdfcreator  = {LaTeX with hyperref},
  pdfdisplaydoctitle = true
}

\title{\Large\bfseries{Quantitative Tsirelson's Theorems via Approximate Schur's Lemma and Probabilistic Stampfli's Theorems}}

\author[1,2,3*]{Xiangling Xu}
\author[1,2,3]{Marc-Olivier Renou}
\author[4,5,6]{Igor Klep}

\affil[1]{Inria Paris-Saclay, B\^atiment Alan Turing, 1 rue Honor\'e d'Estienne d'Orves, 91120 Palaiseau, France}
\affil[2]{CPHT, Ecole polytechnique, Institut Polytechnique de Paris, Route de Saclay, 91128 Palaiseau, France}
\affil[3]{LIX, Ecole polytechnique, Institut Polytechnique de Paris, Route de Saclay, 91128 Palaiseau, France}
\affil[4]{Faculty of Mathematics and Physics, University of Ljubljana}
\affil[5]{Faculty of Mathematics, Natural Sciences
and Information Technologies, 
University of Primorska}
\affil[6]{Institute of Mathematics, Physics and Mechanics, Ljubljana, Slovenia}
\affil[*]{\scriptsize\texttt{xu.xiangling@inria.fr}}

\date{\vspace{-3.00em}}

\begin{document}
\maketitle

\blfootnote{%
}


\begin{abstract}
Whether an almost-commuting pair of operators must be close to a commuting pair is a central question in operator and matrix theory.
We investigate this problem for pairs of $C^*$-subalgebras $\cA$ and $\cB$ of $M_d(\mathds{C})$, showing that each operator in $\cB$ is $O(d^2\epsilon)$-close in operator norm to an operator in the commutant $\cA'$ under two complementary formulations of ``$\epsilon$-almost commutation.''
One formulation is probabilistic, requiring that the operators of $\cB$ have small commutators for most Haar-random unitaries acting on $\cA$.
This first formulation leads to two novel probabilistic generalizations of Stampfli's theorem, which relates an operator's distance from the scalars to the norm of its inner derivation.
The second formulation is deterministic, requiring small commutators between the generators of $\cA$ and $\cB$; we analyze this using an approximate Schur's lemma formulated in terms of Weyl-Heisenberg (clock-and-shift) matrices.
As an application of our results to quantum information theory, we obtain a quantitative Tsirelson's theorem: in dimension $d$, every $\epsilon$-almost quantum commuting observable model is well approximated by a quantum tensor-product model with error $O(d^2\epsilon)$.
\end{abstract}

\section{Introduction}
A long-standing research theme in matrix theory asks when pairs of matrices that \emph{almost} commute must be close to a genuinely commuting pair.
Initiated by~\cite{rosenthal1969almost, halmos1976some} for the normalized Hilbert-Schmidt norm and the operator norm, respectively, many results have been developed for these two norms under various structural assumptions~\cite{luxemburg1970almost,bastian1974subnormal,pearcy1979almost,voiculescu1983asymptotically,choi1988almost,lin1996almost,friis1996almost,glebsky2010almost,filonov2010hilbert,ioana2024almost,lin2024almost}.
(See Sec.~\ref{sec:RelateToPreviousAlmostCommuting} for more details.)
Our perspective and setup are motivated by quantum information theory: Tsirelson's theorem~\cite{tsirelson2006bell} establishes an equivalence between the tensor-product model and the commuting operator model in finite-dimensional quantum systems.
This naturally prompts the quantitative question: to what extent does the equivalence persist when the operators only approximately commute?

\subsection{Contribution and structure of the paper}
Let $\cA,\cB\subset M_d(\mathds{C})$ be finite-dimensional $C^*$-subalgebras with self-adjoint generating families $\{\A\}_{a,x}$ and $\{\B\}_{b,y}$, respectively.
Under an ``$\epsilon$-almost commuting'' hypothesis for the generators---formulated in two complementary senses made precise in Thms.~\ref{thm:ApproximateTsirelsonGeneral} and~\ref{thm:ApproximateTsirelsonProbabilistic}---we show that each generator $\B$ of $\cB$ is $O(d^2\epsilon)$-close (in operator norm) to a self-adjoint operator $\B'$ lying in the commutant $\cA'$.
In finite dimensions, semi-simplicity of $\cA$ implies that $\cA \simeq \bigoplus_{l=1}^L B(\cH_A^l) \otimes \id_B^l$, for some Hilbert spaces $\cH_A^l$ and $\cH_B^l$ with $\id_B^l$ acting on $\cH_B^l$ as the identity.
Consequently, $\B' \in \cA'$ is equivalent to $\B' = \bigoplus_l \id_A^l \otimes \B'^l$, so our bounds yield an explicit approximate tensor-product factorization for the pair $(\cA,\cB)$.

The first route, detailed in Sec.~\ref{sec:ClockAndShiftSection}, is deterministic.
This approach controls commutator bounds with respect to the Weyl-Heisenberg clock-and-shift matrices, which form a canonical generating set for $M_d(\mathds{C})$.
The key technical tool is an approximate Schur's lemma (Lem.~\ref{lem:ApproxSchur}), which quantitatively shows that an operator almost commuting with the clock-and-shift matrices must be close to a scalar, and its consequence in bipartite systems (Lem.~\ref{lem:BipartitionApproxSchur}).
Under assumptions on how efficiently the clock-and-shift matrices can be expressed using the generators $\{\A\}$ of $\cA$ (quantified by algebraic complexity constants $c_1, c_2, c_3$), this leads to our first approximate Tsirelson's theorem (Thm.~\ref{thm:ApproximateTsirelsonGeneral}), establishing the $O(d^2\epsilon)$ estimate in the operator norm.
We also discuss the scaling of the constants $c_1, c_2, c_3$ (Rem.~\ref{rem:ScalingFactors}).

Sec.~\ref{sec:HaarRandomUnitarySection} presents the second, probabilistic route.
Here, we relax the requirement of uniform commutator bounds. We instead demand small commutators only for most Haar-random unitaries within two-dimensional subspaces (aka, single-qubits), and explain how this can be generalized to arbitrary $d$-dimensional subspaces  in Rem.~\ref{rem:CommentOnKisDimofd}.
This leads to a probabilistic Stampfli's theorem (Thm.~\ref{thm:ProbabilisticStampfli}), relating these probabilistic commutator bounds to an operator's distance from scalar multiples of the identity.
We further extend this to a doubly probabilistic version (Thm.~\ref{thm:DoublyProbabilisticStampfli}) by also randomizing the two-dimensional subspaces.
While the former (Thm.~\ref{thm:ProbabilisticStampfli}) holds for possibly infinite-dimensional Hilbert spaces, the doubly probabilistic variant (Thm.~\ref{thm:DoublyProbabilisticStampfli}) holds only in finite dimensions.
The doubly probabilistic Stampfli's theorem allows us to derive our second main approximate Tsirelson's theorem (Thm.~\ref{thm:ApproximateTsirelsonProbabilistic}), again achieving an $O(d^2\epsilon)$ error estimate.

Finally, Sec.~\ref{sec:ApplicationOutlook} explores the applications and context of our work.
We detail the construction of an approximating tensor-product quantum strategy based on our main theorems (Prop.~\ref{prop:TensorStrategyConstruction}, see the following subsection).
Furthermore, we explore the interplay between our results, the Navascu\'es-Pironio-Ac\'in (NPA) hierarchy~\cite{navascues2008convergent, pironio2010convergent}---a convergent semidefinite programming (SDP) hierarchy characterizing commuting quantum correlations---and computational complexity (Rem.~\ref{rem:ConnectionToNPAComplexity}).
We situate our results within the literature on approximating almost commuting matrices (Sec.~\ref{sec:RelateToPreviousAlmostCommuting}) and conclude with a broader discussion (Sec.~\ref{sec:FinalDiscussion}).

\subsection{Motivation: application to quantum information theory}\label{sec:QITmotivation}
Quantum information theory studies how information is represented and processed in quantum systems.
In the finite-dimensional setting relevant here, one works with states (density matrices) and observables (self-adjoint operators/POVMs) on Hilbert spaces and asks how structural assumptions---such as subsystem independence---constrain observable correlations.

Quantum information theory offers two different axioms for composing independent subsystems.
The standard composition axiom~\cite{nielsen2010quantum} postulates that the Hilbert space $\cH_T$ describing a joint system $T=\{A,B\}$ is the tensor product $\cH_T:=\cH_A\otimes\cH_B$.
Each local subsystem's observables act on its respective tensor factors and are identity over the other subsystem: e.g., unitaries $U_A, V_B$ local to subsystems $A, B$ act on the global system $T$ as $U_A \otimes \id_B$ and $\id_A \otimes V_B$, respectively.
An alternative axiomatization, common in algebraic quantum field theory~\cite{landsman2017foundations}, does not introduce the tensor product, but models independence by postulating that the observables of different parties act on the same Hilbert space and commute.
In other words, $U_A, V_B$ act on the same global Hilbert space $\cH_T$ with only $[U_A, V_B]=0$.

Whether these two axiomatizations result in the same physical predictions was a long-standing problem.
In particular, Tsirelson's problem asks whether the \emph{tensor-product model} and the \emph{commuting operator model} yield the same set of bipartite quantum correlations---a concept that has been central to both the foundational understanding and practical applications of quantum theory since Bell's groundbreaking work~\cite{bell1964einstein}.
Tsirelson demonstrated the equivalence of these two models in finite-dimensional settings~\cite{tsirelson2006bell, scholz2008tsirelson, doherty2008quantum}, formally:
\begin{theorem*}[Tsirelson's theorem]\label{thm:TsirelsonTheoremOG_Structural}
    Let $\{\A\}, \{\B\} \subset B(\cH)$ be two sets of positive operator-valued measures (POVMs) on a finite-dimensional Hilbert space $\cH$ such that $[A_{a|x},B_{b|y}]=0$ for all $a,b,x,y$.
    Then there exists a decomposition of $\cH \simeq \bigoplus_l (\mathcal{H}_A^l \otimes \mathcal{H}_B^l)$, such that operators $\A$ and $\B$ take the form
    \begin{align*}
        \A = \bigoplus_l (\A^l \otimes \id_B^l) \quad \text{and} \quad \B = \bigoplus_l (\id_A^l \otimes \B^l),
    \end{align*}
    where $\A^l \in B(\mathcal{H}_A^l)$ and $\B^l \in B(\mathcal{H}_B^l)$.
    Consequently, any commuting-operator correlations
    \begin{align*}
        p(ab|xy) = \Tr{\rho \cdot \A \cdot \B}
    \end{align*}
    obtained from a state $\rho$ on $\mathcal{H}$ can be reproduced as tensor-product correlations using POVMs $\Tilde{A}_{a|x} \in B(\bigoplus_l \mathcal{H}_A^l)$ and $\Tilde{B}_{b|y} \in B(\bigoplus_l \mathcal{H}_B^l)$ on a state $\Tilde{\rho}$ over $(\bigoplus_l \mathcal{H}_A^l) \otimes (\bigoplus_l \mathcal{H}_B^l)$, i.e.,
    \begin{align*}
        p(ab|xy) = \Tr{ \Tilde{\rho} \cdot (\Tilde{A}_{a|x} \otimes \Tilde{B}_{b|y}) }.
    \end{align*}
    The statement extends inductively to multipartite cases.
\end{theorem*}

Tsirelson conjectured that the above theorem can be generalized to infinite dimensions.
However, this was recently shown to be false by the seminal work~\cite{ji2021mip}: the commuting operator model can produce correlations unattainable by any tensor-product quantum strategy.
This result has far-reaching implications including a disproof of Connes' embedding conjecture~\cite{connes1976classification} and Kirchberg's conjecture~\cite{kirchberg1993non}; see~\cite{ozawa2013connes} for a nice survey.

By contrast, in finite dimensions, our results lead to a quantitative robustness version of Tsirelson's theorem: if a $d$-dimensional bipartite strategy is $\epsilon$-almost commuting in the sense of Thms.~\ref{thm:ApproximateTsirelsonGeneral} and~\ref{thm:ApproximateTsirelsonProbabilistic}, then its correlations are approximated by those of a genuine tensor-product strategy with error $O(d^2\epsilon)$ (Prop.~\ref{prop:TensorStrategyConstruction}).
This is a quantitative counterpart to Ozawa's asymptotic result \cite{ozawa2013tsirelson}, further demonstrating that, in finite dimensions, the tensor-product model remains a sound effective description when subsystem independence holds only approximately.
The multipartite extension follows by induction, as in the original Tsirelson's theorem.

\section{Weyl-Heisenberg (clock-and-shift) formulation}\label{sec:ClockAndShiftSection}
In this section, we present our first main result, a deterministic version of Tsirelson's theorem in Thm.~\ref{thm:ApproximateTsirelsonGeneral}.
Recall the key ideas for proving Tsirelson's theorem (see, e.g.~\cite[App.~A]{doherty2008quantum}):
any finite-dimensional $C^*$-algebra $\cA$ generated by Alice's observables $A_{a|x}$ decomposes as a direct sum of simple blocks $\cA \simeq \bigoplus_l B(\cH_A^l) \otimes \id_B^l$, and by Schur's lemma, Bob's commuting algebra $\cB$ must lie in $\bigoplus_l \id_A^l \otimes B(\cH_B^l)$.

Our approximate analogue replaces ``commutes'' by ``almost commutes''.
In particular, to formulate our results, we pick Sylvester's \emph{clock} $\Sigma_3$ and \emph{shift} $\Sigma_1$ unitaries as generators: almost commuting with this pair already controls an operator in every direction.

With these generators we prove an approximate Schur's lemma (Lem.~\ref{lem:ApproxSchur}) and its bipartite version (Lem.~\ref{lem:BipartitionApproxSchur}).
Next, we show an approximate Tsirelson's theorem for simple algebras (Lem.~\ref{lem:ApproximateTsirelsonSimple}), and then by the same block-decomposition argument, a general approximate Tsirelson's theorem (Thm.~\ref{thm:ApproximateTsirelsonGeneral}).
We finish with a discussion of scaling (Rem.~\ref{rem:ScalingFactors}).

\subsection{Clock-and-shift matrices}\label{sec:ClockAndShift}
Recall Sylvester's clock-and-shift matrices~\cite{appleby2005symmetric}, which generalize the Pauli matrices to a $d$-dimensional Hilbert space $\cH \simeq \mathds{C}^d$.
Also known as the Weyl-Heisenberg matrices, they are fundamental in finite-dimensional quantum mechanics due to their connection to Weyl's formulation of the canonical commutation relations.
These matrices serve as analogs of position and momentum operators in finite-dimensional quantum systems.

Let $\omega = e^{2\pi i/d}$ be the $d$th root of unity.
Using Dirac's notation, denote by $\{\ket{i} \mid i = 0, \dots, d-1\}$ the standard basis of $\cH$, and $\ket{i+j}$ is understood up to mod $d$.
Then the shift matrix $\Sigma_1 \in B(\cH)$ is defined by $\Sigma_1: \ket{i} \mapsto \ket{i+1}$ and the clock matrix is defined by $\Sigma_3: \ket{i} \mapsto \omega^{i} \ket{i}$. 
More explicitly:
\begin{equation}
    \begin{aligned}
        \Sigma_1 &= \begin{pmatrix}
                0 & 0 & \cdots & 0 & 1 \\
                1 & 0 & \cdots & 0 & 0 \\
                0 & 1 & \cdots & 0 & 0 \\
                \vdots & \vdots & \ddots & \vdots & \vdots \\
                0 & 0 & \cdots & 1 & 0
            \end{pmatrix}, \quad \Sigma_3 = \begin{pmatrix}
                1 & 0 & \cdots & 0 \\
                0 & \omega & \cdots & 0 \\
                \vdots & \vdots & \ddots & \vdots \\
                0 & 0 & \cdots & \omega^{d-1}
            \end{pmatrix}.
    \end{aligned}
\end{equation}
The notation comes from the fact that the shift matrix $\Sigma_3$ (resp. clock matrix $\Sigma_1$) is a generalization of the Pauli $Z$-matrix $\sigma_3$ (resp. $X$-matrix $\sigma_1$) when $d=2$.
The clock and shift $\Sigma_3, \Sigma_1$ satisfy a generalized algebraic relation of the Pauli matrices in the sense that
\begin{equation}
    \begin{aligned}
        \Sigma_1^d = \Sigma_3^d = \id, \\
        \Sigma_3 \Sigma_1 = \omega \Sigma_1 \Sigma_3.
    \end{aligned}
\end{equation}
Also note that both $\Sigma_1, \Sigma_3$ are unitary and traceless, but no longer Hermitian when $d > 2$.
Lastly, they give rise to an orthogonal basis of $B(\cH)$ (w.r.t Hilbert-Schmidt inner product) composed of unitary matrices
\begin{align}
    \{ \sigma_{k, l} := \Sigma_1^k \Sigma_3^l = \sum_{j=0}^{d-1} \omega^{jl} \ketbra{j+k}{j} \}_{0 \leq k,l \leq d-1},
\end{align}
where $\ketbra{j+k}{j} := (\ket{j+k})^* \ket{j}$.

\subsection{Approximate Schur's Lemma and its bipartite version}
Elementary linear algebraic arguments lead to an approximate version of Schur's Lemma for finite-dimensional Hilbert spaces.
Write $\lVert \cdot \rVert_{\op}$ for the operator (spectral) norm and $\lVert \cdot \rVert_{\max}$ for the max norm, i.e., $\lVert A \rVert_{\max} = \max_{i,j} \lvert A_{ij} \rvert$.

\begin{lemma}[Approximate Schur's Lemma]\label{lem:ApproxSchur}
    Let $\cH$ be a $d$-dimensional Hilbert space, $d < \infty$.
    Consider a fixed matrix $C \in B(\cH) \simeq M_d(\mathds{C})$ and suppose there exists an $\epsilon > 0$ such that, for both $i=1, 3$,
    \begin{align}
        \lVert [C, \Sigma_i] \rVert_{\op} \leq \epsilon.
    \end{align}
    Then there exists $c \in \mathds{C}$ such that
    \begin{align}
        \lVert C - c\id \rVert_{\op} \leq (d-1)\epsilon.
    \end{align}
\end{lemma}
\begin{proof}
    Note that for any matrices $X, Y, Z$ we have
    \begin{align*}
        \lVert [X Y, Z] \rVert_{\op} = \lVert X[Y, Z] + [X, Z]Y \rVert_{\op} \leq \lVert X \rVert_{\op} \lVert [Y, Z] \rVert_{\op} + \lVert [X, Z] \rVert_{\op} \lVert Y \rVert_{\op}.
    \end{align*}
    Then, by induction and the fact that $\lVert \Sigma_i \rVert_{\op} =1$, we have
    \begin{align*}
        \lVert [C, \Sigma_i^k] \rVert_{\op} \leq k\epsilon.
    \end{align*}

    Next, note that $C \Sigma_1$ is the matrix where each column of $C$ is cyclically shifted leftward, and $\Sigma_1 C$ is the matrix where each row of $C$ is cyclically shifted downward.
    That is, for all $i$, the $(i, i-1)$-entry of $C \Sigma_1$ is $C_{ii}$, while the $(i, i-1)$-entry of $(\Sigma_1 C)$ is $C_{i-1, i-1}$.
    Then the assumption imposes that
    \begin{align*}
        \lvert C_{i,i} - C_{i-1, i-1} \rvert = \vert (C \Sigma_1 -\Sigma_1 C)_{i, i-1} \rvert \leq \lVert C \Sigma_1 - \Sigma_1 C \rVert_{\max} \leq \lVert [C, \Sigma_1] \rVert_{\op} \leq \epsilon,
    \end{align*}
    Consequently,
    \begin{align*}
        \lvert C_{ii} - \frac{1}{d}\Tr{C}\rvert = \frac{1}{d} \lvert \sum_j (C_{ii} - C_{jj}) \rvert \leq \frac{1}{d} ( 0 + \epsilon + \cdots + (d-1)\epsilon) = \frac{d-1}{2}\epsilon
    \end{align*}
    by repeated uses of the triangle inequality.

    Moreover, observe that $\diag(C) = 1/d\sum_{k = 0}^{d-1} \Sigma_3^k C \Sigma_3^{-k}$, since
    \begin{align*}
        ( \frac{1}{d} \sum_{k = 0}^{d-1} \Sigma_3^k C \Sigma_3^{-k})_{ij}= \frac{1}{d} \underbrace{\sum_{k = 0}^{d-1} \omega^{k(i-j)}}_{d \delta_{i,j}} C_{ij} = \delta_{i,j} C_{ij}.
    \end{align*}
    It follows that
    \begin{align*}
        \lVert C - \diag(C) \rVert_{\op} &= \frac{1}{d} \lVert \sum_{k = 0}^{d-1} (C - \Sigma_3^k C \Sigma_3^{-k}) \rVert_{\op} \leq \frac{1}{d} \sum_{k = 0}^{d-1} \lVert [\Sigma_3^k, C] \Sigma_3^{-k} \rVert_{\op} \\
        &\leq \frac{1}{d} \sum_{k = 0}^{d-1} \lVert [\Sigma_3^k,C] \rVert_{\op} \lVert \Sigma_3^{-k} \rVert_{\op} \leq \frac{1}{d} \sum_{k = 0}^{d-1} k\epsilon \cdot 1^k = \frac{d-1}{2}\epsilon.
    \end{align*}
    Finally, let $c = 1/d \Tr{C}$, we compute
    \begin{align*}
        \lVert C - c\id \rVert_{\op} &\leq \lVert C - \diag(C) \rVert_{\op} + \lVert \diag(C) - \frac{1}{d}\Tr{C}\id \rVert_{\op} \\
        &\leq \frac{d-1}{2}\epsilon + \max_{i} \lvert C_{ii} - \frac{1}{d}\Tr{C} \rvert \leq \frac{d-1}{2}\epsilon + \frac{d-1}{2}\epsilon = (d-1) \epsilon.
    \end{align*}
\end{proof}

Via manipulation of the Kronecker tensor product formula, we quickly obtain a bipartite version of approximate Schur's Lemma.
\begin{lemma}\label{lem:BipartitionApproxSchur}
    Consider two Hilbert spaces $\cH_1$ with dimension $d_1$ and $\cH_2$ with dimension $d_2$.
    Suppose for the matrix $C \in B(\cH_1 \otimes \cH_2)$, there exists some $\epsilon > 0$ such that
    \begin{align}
        \lVert [C, \id_1 \otimes \Sigma_1] \rVert_{\op},\, \lVert [C, \id_1 \otimes \Sigma_3] \rVert_{\op} \leq \epsilon.
    \end{align}
    Then the matrix $C' = 1/d_2\tr{\cH_2}{C} \in B(\cH_1)$, where $\mathrm{Tr}_{\cH_2}$ denotes the partial trace $B(\cH_1 \otimes \cH_2) \to B(\cH_1)$, satisfies
    \begin{align}
        \Vert C - C' \otimes \id_2 \rVert_{\op} \leq d_1 d_2^2 \epsilon.
    \end{align}
    In addition, if $C$ is positive semidefinite then so is $C'$.
\end{lemma}
\begin{proof}
    Note that
    \begin{align*}
        C = \begin{pmatrix} 
                C_{(11)} & \cdots  & C_{(1 d_1)}\\
                \vdots & \ddots & \vdots\\
                C_{(d_1 1)} & \cdots  & C_{(d_1 d_1)} 
            \end{pmatrix},
    \end{align*}
    where each $C_{(ij)}$ is some $d_2 \times d_2$ matrix in $B(\cH_2)$.
    Similarly, we have that
    \begin{align*}
        \id_1 \otimes \Sigma_i = \begin{pmatrix} 
                \Sigma_i &    &  \\
                  & \ddots &  \\
                  &    & \Sigma_i 
            \end{pmatrix}
    \end{align*}
    are block matrices with only $\Sigma_i$ on the diagonal.
    Since the operator norm of a matrix upper-bounds the operator norm of its blocks, the condition $\lVert [C, \id_1 \otimes \Sigma_i] \rVert_{\op} \leq \epsilon$ implies that, for all $k, l$,
    \begin{align*}
        \lVert [C_{(kl)}, \Sigma_1] \rVert_{\op}, \, \lVert [C_{(kl)}, \Sigma_3] \rVert_{\op} \leq \epsilon.
    \end{align*}

    Then, applying the approximate version of Schur's Lemma~\ref{lem:ApproxSchur}, for each $k, l$ we check $c_{kl} := 1/d_2 \Tr{C_{(kl)}}$ satisfies
    \begin{align*}
        \lVert C_{(kl)} - c_{kl}\id_2 \rVert_{\max} \leq \lVert C_{(kl)} - c_{kl}\id_2 \rVert_{\op} \leq (d_2-1) \epsilon \leq d_2 \epsilon.
    \end{align*}
    Defining $C' = ( c_{kl} ) \in B(\cH_1)$, it follows that
    \begin{align*}
        C' \otimes \id_2 = \begin{pmatrix} 
                c_{11}\id_2 & \cdots  & c_{1 d_1}\id_2\\
                \vdots & \ddots & \vdots\\
                c_{d_1 1}\id_2 & \cdots  & c_{d_1 d_1}\id_2
            \end{pmatrix}.
    \end{align*}
    Hence,
    \begin{align*}
        \Vert C - C' \otimes \id_2 \rVert_{\op} \leq d_1 d_2 \Vert C - C' \otimes \id_2 \rVert_{\max} \leq d_1 d_2 \max_{kl} (\lVert C_{(kl)} - c_{kl}\id_2 \rVert_{\max}) \leq d_1 d_2^2 \epsilon.
    \end{align*}
    (Note that the above operator norm upper bound of $C - C' \otimes \id_2$ via its block is the tightest general bound, e.g., consider the matrix whose entries are all $1$ and take each entry as a block.)

    Lastly, observe that $C'$ is in fact the normalized partial trace of $C$, since
    \begin{align*}
        \tr{\cH_2}{C} = \begin{pmatrix} 
                \Tr{C_{(11)}} & \cdots  & \Tr{C_{(1 d_1)}}\\
                \vdots & \ddots & \vdots\\
                \Tr{C_{(d_1 1)}} & \cdots  & \Tr{C_{(d_1 d_1)}} 
            \end{pmatrix}
            = d_2\begin{pmatrix} 
                c_{11} & \cdots  & c_{1 d_1}\\
                \vdots & \ddots & \vdots\\
                c_{d_1 1} & \cdots  & c_{d_1 d_1} 
            \end{pmatrix} = d_2 C'.
    \end{align*}
    This implies that if $C$ is positive semidefinite, then so is $C'$, due to complete positivity of the partial trace~\cite[Ch.~II.6.10]{blackadar2006operator}
\end{proof}

\subsection{Approximate Tsirelson's theorem from clock-and-shift matrices}
Before presenting the approximate version of Tsirelson's theorem, we recall the Artin-Wedderburn decomposition of a finite-dimensional $C^*$-algebra~\cite[II.1.6.4, II.8.3.2(iv), and III.1.5.3]{blackadar2006operator}. 
\begin{lemma}\label{lem:SemiSimpleFinDim}
    Every finite-dimensional $C^*$-algebra $\cA$ is semi-simple.
    That is, there exists an Artin-Wedderburn decomposition
    \begin{align*}
        \cA = \bigoplus_k \cA_k,
    \end{align*}
    such that each $\cA_k$ is simple, i.e. contains no non-trivial closed two-sided ideals.

    Furthermore, if $\cA \subset B(\cH)$ is simple, then there exists a bipartite partition of $\cH$ such that $\cH = \cH_1 \otimes \cH_2$ and $\cA \simeq B(\cH_1) \otimes \id_2$.
\end{lemma}

This structural result gives a road map: first prove the simple version of approximate Tsirelson's theorem, then the general case follows.
For the following approximate Tsirelson's theorem, let us impose extra assumptions on the ``generating power of the strategy'', represented by the constants $c_1, c_2, c_3$ below.
\begin{lemma}[Approximate Tsirelson's theorem, simple case]\label{lem:ApproximateTsirelsonSimple}
Let $\cH$ be a $d$-dimensional Hilbert space $\cH$.
    Let $\cA\subset B(\cH)$ be generated by contractive self-adjoint operators $\{ \A \} $ and $\cB\subset B(\cH)$ be generated by contractive self-adjoint operators $\{ \B \}$.
    Assume that there exists an $\epsilon > 0$, such that for all $a, b, x, y$,
    \begin{align}
        \lVert [\A, \B] \rVert_{\op} \leq \epsilon.
    \end{align}
    Suppose that $\cA$ is simple, i.e. there exists a bipartition $\cH = \cH_A \otimes \cH_B$ such that $\cA \simeq B(\cH_A) \otimes \id_B$ and $\A = \A' \otimes \id_B$ for all $a, x$.
    
    Suppose that the clock-and-shift matrices $\Sigma_3, \Sigma_1 \in B(\cH_A)$ are generated by some polynomials $P_3, P_1$ in $\{\A'\}$.
    Assume, moreover, that the maximal absolute value of their coefficients is bounded by $c_1$, the maximal degree is bounded by $c_2$, and the maximal number of terms is bounded by $c_3$.
    Then there exist operators $\B' \in B(\cH_B)$ such that, for all $b, y$,
    \begin{align}
        \lVert \B - \id_A \otimes \B' \rVert_{\op} \leq c_1 c_2 c_3 d^2 \epsilon.
    \end{align}
    In addition, if $\B$ is positive then so is $\B'$
\end{lemma}
\begin{proof}
    Note that for any matrices $X, Y, Z$ we have
    \begin{align}\label{eq:ProductCommutatorFormula}
        \lVert [X Y, Z] \rVert_{\op} = \lVert X[Y, Z] + [X, Z]Y \rVert_{\op} \leq \lVert X \rVert_{\op} \lVert [Y, Z] \rVert_{\op} + \lVert [X, Z] \rVert_{\op} \lVert Y \rVert_{\op}.
    \end{align}
    Then for any monomial $\alpha$ in $\{ \A \}$ of degree $k$, one can use the fact that $\lVert \A \rVert_{\op} \leq 1$ to inductively compute
    \begin{align*}
        \lVert [\alpha, \B] \rVert_{\op} \leq k \max_{a, x} \lVert \A \rVert_{\op} \lVert [\A, \B] \rVert_{\op} \leq k \epsilon.
    \end{align*}
    Then, for polynomials $\Sigma_i = P_i(\{\A\})$, we have
    \begin{align*}
        \lVert [\Sigma_i, \B] \rVert_{\op} \leq c_1 c_2 c_3 \max_{a, x} (\lVert [\A, \B] \rVert_{\op}) \leq c_1 c_2 c_3 \epsilon,
    \end{align*}
    and we are done by Lem.~\ref{lem:BipartitionApproxSchur}.
\end{proof}

Remark that the ``contraction'' requirement is not necessary and one can reproduce the same result by replacing $\epsilon$ by $\epsilon/\lVert \A \rVert_{\op}$.
The simple version can be readily generalized to the  general finite-dimensional case with Lem.~\ref{lem:SemiSimpleFinDim}.

\begin{theorem}[Approximate Tsirelson's theorem, general case]\label{thm:ApproximateTsirelsonGeneral}
    Let $\cA$ be generated by contractive self-adjoint operators $\{ \A \} \subset B(\cH)$ and $\cB$ be generated by contractive self-adjoint operators $\{ \B \} \subset B(\cH)$ for some $d$-dimensional Hilbert space $\cH$.
    Assume that there exists an $\epsilon > 0$, such that for all $a, b, x, y$,
    \begin{align}
        \lVert [\A, \B] \rVert_{\op} \leq \epsilon.
    \end{align}
    Suppose also that $\cA$ admits the Artin-Wedderburn decomposition
    \begin{align*}
        \cA = \bigoplus_{l=1}^L \cA_l \simeq \bigoplus_{l=1}^L B(\cH_A^l) \otimes \id_B^l \text{ and } \A = \bigoplus_{l=1}^L \A^l \otimes \id_B^l,
    \end{align*}
    with the corresponding orthogonal projectors $\Pi_l$ to the direct summands.
    Denote by $\Sigma_3^l, \Sigma_1^l \in B(\cH_A^l)$ the clock-and-shift operators in $\cH_A^l$.

    Furthermore, suppose that there exist polynomials $P_l, Q_1^l, Q_3^l$, for all $l=1, \dots, L$ such that
    \begin{align*}
        \Pi_l = P_l(\{\A\}), \, \Sigma_1^l = Q_1^l(\{\Pi_l \A\ \Pi_l\}), \text{and } \Sigma_3^l = Q_3^l(\{\Pi_l \A\ \Pi_l\}).
    \end{align*}
    Assume that their absolute values of the maximal coefficients are bounded by the constant $c_1$, the degrees are bounded by the constant $c_2$, and the maximal number of terms is bounded by the constant $c_3$.
    Then there exist operators $\B' \in \bigoplus_{l=1}^L \id_A^l \otimes B(\cH_B^l) = \cA'$ such that, for all $b, y$,
    \begin{align}\label{eq:ApproximateTsirelsonBound}
        \lVert \B - \B' \rVert_{\op} \leq 2 c_1 c_2 c_3 \left( c_1 c_2 c_3 + 1 \right)d^2 \epsilon.
    \end{align}
    In addition, if $\B$ is positive then so is $\B'$.
\end{theorem}
\begin{proof}
    First, we wish to apply Lem.~\ref{lem:ApproximateTsirelsonSimple} to each $\Pi_l \A \Pi_l \in \cA_l$ and the corresponding $\Pi_l \B \Pi_l$.
    To this end, note that, by straightforward calculations, one has
    \begin{align*}
        \lVert [\B, \Pi_l] \rVert_{\op} \leq c_1 c_2 c_3 \lVert [\A, \B] \rVert_{\op} \leq c_1 c_2 c_3 \epsilon
    \end{align*}
    and
    \begin{align*}
        [ \Pi_l \A \Pi_l, \Pi_l \B \Pi_l] = \Pi_l [\Pi_l, B] \A \Pi_l + \Pi_l [\A, \B] \Pi_l + \Pi_l \A [\Pi, \B] \Pi_l.
    \end{align*}
    It follows from $\lVert \A \rVert_{\op}, \lVert \Pi_l \rVert_{\op} \leq 1$ and the Cauchy-Schwarz inequality that
    \begin{align*}
        \lVert [ \Pi_l \A \Pi_l, \Pi_l \B \Pi_l] \rVert_{\op} \leq \lVert [\Pi_l, B]  \rVert_{\op} + \lVert [\A, \B] \rVert_{\op}+ \lVert [\Pi, \B] \rVert_{\op} \leq (2 c_1 c_2 c_3 + 1) \epsilon.
    \end{align*}
    Therefore, by Lem.~\ref{lem:ApproximateTsirelsonSimple}, there exists for each $l$ some positive semidefinite $\B^l \in B(\cH_B^l)$ such that
    \begin{align*}
        \lVert \Pi_l \B \Pi_l - \id_A^l \otimes \B^l \rVert_{\op} \leq c_1 c_2 c_3 (2 c_1 c_2 c_3 + 1) d^2_l \epsilon,
    \end{align*}
    where $d_l = \dim(\cH_A^l \otimes \cH_B^l)$.

    Now, $\B$ does not admit the same direct decomposition as $\cA$ due to $[\Pi_l, \B] \neq 0$. 
    Thus, we need to also estimate
    \begin{align*}
        \lVert \B - \sum_l \Pi_l \B \Pi_l \rVert_{\op} & \leq \lVert \sum_{l, l'} \Pi_l \B \Pi_{l'} - \sum_l \Pi_l \B \Pi_l \rVert_{\op} \\
        & \leq \lVert \sum_{l \neq l'} \Pi_l \B \Pi_{l'} \rVert_{\op} \\
        & \leq \sum_{l \neq l'} \lVert  \Pi_l \Pi_{l'} \B + \Pi_l [\B, \Pi_{l'}] \rVert_{\op} \\
        & \leq \sum_{l \neq l'} \lVert \Pi_l \rVert_{\op} \lVert [\B, \Pi_{l'}] \rVert_{\op} \leq L(L-1) c_1 c_2 c_3 \epsilon,
    \end{align*}
    where the completeness and orthogonality of $\Pi_l$ are used.
    
    Finally, one sees that
    \begin{align*}
        \lVert \B - \bigoplus_{l=1}^L \id_A^l \otimes \B^l \rVert_{\op} \leq \lVert \B - \sum_{l=1}^L \Pi_l \B \Pi_l \rVert_{\op} + \lVert \sum_{l=1}^L\Pi_l \B \Pi_l - \bigoplus_{l=1}^L\id_A^l \otimes \B^l \rVert_{\op} \\
         \leq L(L-1) c_1 c_2 c_3 \epsilon + \sum_{l=1}^L c_1 c_2 c_3 (2 c_1 c_2 c_3 +1) d^2_l \epsilon \leq  c_1 c_2 c_3 \left( L(L-1) + (2 c_1 c_2 c_3 + 1)d^2 \right) \epsilon.
    \end{align*}
    Note $L \leq d$, so we are done by defining $\B' := \bigoplus_{l=1}^L \id_A^l \otimes \B^l$.
\end{proof}

While all results in this section are formulated in terms of the Weyl-Heisenberg (clock-and-shift) matrices $\Sigma_1, \Sigma_3$, we observe that any full generating set (e.g., the matrix units $E_{kl} = \ketbra{k}{l}$) would still work.
We finish the section with a remark on the factors $c_1, c_2, c_3$.

\begin{remark}\label{rem:ScalingFactors}
    Note that the bound Eq.~\eqref{eq:ApproximateTsirelsonBound} has an $O(d^2\epsilon)$ scaling for fixed $c_1, c_2, c_3$.
    We finish the section with some comments on these factors $c_1, c_2, c_3$, 
    which are generally example-specific.
    \begin{enumerate}
        \item The generating polynomial degree $c_2$ is related to the length of algebras with known dependence on the dimension $d$.
        The conjectured bound is $O(d)$ according to~\cite{paz1984application}, while the best proven bound is $O(d \log(d))$ due to~\cite{shitov2019improved}.
        The bound is $O(\log d)$ when a ``generic'' assumption is met as detailed in~\cite{klep2016sweeping}.
        
        \item The number of terms in generating polynomials $c_3$ is related to $c_2$.
        In the worst case scenario, $c_3$ is the number of possible monomials of $\{ \A \}$ up to degree $c_2$, which grows exponentially in $c_2$ (i.e. $c_3 \leq \sum_{k=0}^{c_{2}} ( \lvert\{ \A \}\rvert )^{k}$).
        
        \item The coefficient magnitude $c_1$ can be challenging to bound generally.
        While specific algebraic structures might lead to large $c_1$, work by~\cite{pascoe2019elementary} offers a systematic approach.
        It involves constructing a matrix $P$ from the POVM generators, whose properties (e.g., its singular values or entry magnitudes) can serve as an indicator for the likely behavior of $c_1$.
    \end{enumerate}
\end{remark}

\section{Haar-random single-qubit unitary formulation}\label{sec:HaarRandomUnitarySection}
In this section, we present our second main result, a probabilistic version of Tsirelson's theorem in Thm.~\ref{thm:ApproximateTsirelsonProbabilistic}.
Since this will have 
a matrix generator independent formulation, we start by looking for a uniform version of Schur's lemma.
Such a result is given by Stampfli~\cite[Thm.~4 \& Cor.~1]{stampfli1970norm}.
\begin{theorem}[Stampfli's theorem]\label{thm:StampfliTheorem}
    Let $C \in B(\cH)$ for some Hilbert space $\cH$ of possibly infinite dimensions.
    Then
    \begin{align}
        \sup_{B \in B(\cH):\, \lVert B \rVert_{\op} = 1} \lVert [C, B] \rVert_{\op} = \inf_{c \in \mathds{C}} 2 \lVert C - c \id \rVert_{\op}.
    \end{align}
    Moreover, if  $C$ is a normal operator, then
    \begin{align}
        \inf_{c \in \mathds{C}} \lVert C - c \id \rVert_{\op} = R(\sigma(C)),
    \end{align}
    where $R(\sigma(C))$ is the radius of the minimum enclosing disk of the compact set $\sigma(C) \subset \mathds{C}$.
    Note that $R(\sigma(C))$ is not the same as the spectral radius of $C$.
\end{theorem}
In particular, if $\lVert [C, B] \rVert_{\op} \leq \epsilon$ for all $B$ of norm $1$, then $C$ is $\epsilon/2$-close to some scalar operator $c\id$ in operator norm.
Note that $\lVert [C, U] \rVert_{\op} \leq \epsilon$ for all unitaries $U$ is equivalent to the assumption that $\lVert [C, B] \rVert_{\op} \leq \epsilon$ for all $B$ of norm $1$ due to the Russo-Dye theorem~\cite[Cor.~II.3.2.15]{blackadar2006operator}.

While mathematically pleasing, Stampfli's premise is too strong in physical scenarios, since testing commutators with all operators $B$ satisfying $\lVert B \rVert_{\op} = 1$ would require probing an uncountable family of observables.
Therefore, in this section, we revisit Stampfli's theorem through a physically motivated probabilistic approach---the Haar-random single-qubit unitary formulation.
Here, ``single-qubit'' in physics jargon refers to $\mathds{C}^2$-subspaces (more generally, ``qudits'' are $\mathds{C}^d$ for any $d \geq 2$).

We first show that the demanding requirement ``$C$ almost commutes with \emph{every} unitary'' can be relaxed to ``$C$ almost commutes with \emph{most} single-qubit unitaries taken at random'', yielding our probabilistic Stampfli theorem (Thm.~\ref{thm:ProbabilisticStampfli}).
Because the Haar measure is unavailable in infinite dimensions, the randomization is implemented by sampling Haar-random unitaries inside every two-dimensional subspace.
However, checking every two-dimensional subspace in an infinite-dimensional space is still unrealistically demanding.

Hence, we then push the idea further: by also randomizing these two-dimensional subspaces, we obtain a doubly probabilistic Stampfli's theorem (Thm.~\ref{thm:DoublyProbabilisticStampfli}) that is better aligned with realistic experiments.
Though, due to the technicality of randomization over subspaces, our result necessarily restricts to finite dimensions.
Finally, we develop another approximate Tsirelson's theorem (Thm.~\ref{thm:ApproximateTsirelsonProbabilistic}) based on this doubly probabilistic Haar-random single-qubit unitary formulation.

We observe that the above results can also be formulated with $d$-dimensional subspaces for arbitrary $d$ (Rem.~\ref{rem:CommentOnKisDimofd}).

\subsection{Probabilistic Stampfli's theorem}
The first probabilistic relaxation of commutation can be written as follows: for Haar-random unitaries $U$, there are $\epsilon, \delta >0$, such that the probability of having a small commutator ($\leq \epsilon$) with $U$ is high ($\geq 1-\delta$).

For notational convenience, from now on denote by $\mathrm{Gr}(2, \cH)$ the Grassmannian of $\cH$, this is a manifold whose elements are exactly two-dimensional subspaces $\cK \subset \cH$.
Let $\mu_{\cK}$ denote the Haar probability measure on the unitary group $U(\cK)$ in $B(\cK)$.

We first consider self-adjoint operators, and the general case follows from the standard decomposition into real plus imaginary parts.
\begin{lemma}\label{lem:ProbabilisticStampfliNormal}
    Let $\cH$ be a Hilbert space of possibly infinite dimension, and let $C \in B(\cH)$ be a self-adjoint operator.
    Given $\epsilon > 0$ and $\delta \in [0, 1]$, suppose that
    \begin{align}\label{eq:ProbabilisticStampfliAssumption}
        \Pr_{U  \sim \mu_{\cK} }\{ \lVert [P_{\cK} C P_{\cK}, U] \rVert_{\op} \leq \epsilon \} \geq 1 - \delta
    \end{align}
    for every subspace $\cK \in \mathrm{Gr}(2, \cH)$ with projector $P_{\cK}$.
    Then
    \begin{align}
        \inf_{c \in \mathds{C}} \lVert C - c \id_{\cH} \rVert_{\op} \leq \min \left( \frac{\sqrt{2}}{2} \left( \sqrt{1-\delta}\, \epsilon + 2\sqrt{\delta}\lVert C \rVert_{\op} \right),\, \left\lceil{\frac{1}{1-\delta}}\right\rceil \frac{\epsilon}{2} \right).
    \end{align}
\end{lemma}
Note that the bound $\lceil{\frac{1}{1-\delta}}\rceil \frac{\epsilon}{2}$ generally behaves better when $\delta$ is small (e.g., it reduces to $\epsilon$ when $\delta \leq 1/2$), while the bound $\frac{\sqrt{2}}{2} \left( \sqrt{1-\delta}\, \epsilon + 2\sqrt{\delta}\lVert C \rVert_{\op} \right)$ is more stable and does not blow up as $\delta \to 1$.

\begin{proof}
    Let us first derive the bound $\frac{\sqrt{2}}{2} \left( \sqrt{1-\delta}\, \epsilon + 2\sqrt{\delta}\lVert C \rVert_{\op} \right)$.
    The central object to bound is $\mathds{E} \lVert [P_{\cK} C P_{\cK}, U] \rVert_{\op}^2$---while the upper bound is straightforward, the lower bound requires more work.
    The main idea is to identify $\cK = \Span{\ket{\psi_1}, \ket{\psi_2}}$.
    Up to infinite-dimensional subtlety, the vector $\ket{\psi_1}$ (resp. $\ket{\psi_2}$) is chosen to approximate the ``eigenvector'' of $C$ associated with the minimal (resp. maximal) ``eigenvalue''.
    On this $\cK$, one can then lower bound $\lVert [P_{\cK} C P_{\cK}, U] \rVert_{\op}$ by $2\inf_{c \in \mathds{C}} \lVert C - c \id_{\cH} \rVert_{\op}^2$ using the radius of spectrum of $C$, and then apply Stampfli's Theorem~\ref{thm:StampfliTheorem}.

    We begin with the spectral theorem for the bounded self-adjoint operator $C$~\cite[Ch.~I.6.1]{blackadar2006operator}: there exists a unique projective-valued measure $E$ such that
    \begin{align*}
        C = \int_{\sigma(C)} \lambda \, dE(\lambda).
    \end{align*}
    Note that the spectrum $\sigma(C)$ satisfies
    \begin{align*}
        \sigma(C) \subset \left[\inf_{\lVert \psi \rVert_2 = 1}\sandwich{\psi}{C}{\psi} , \sup_{\lVert \psi \rVert_2 = 1}\sandwich{\psi}{C}{\psi} \right] := [\lambda_{\min}, \lambda_{\max}].
    \end{align*}
    We can check in this case that $R(\sigma(C)) = \frac{1}{2}(\lambda_{\max} - \lambda_{\min})$.

    If $\lambda_{\min} = \lambda_{\max}$, then $C$ is automatically a scalar operator and the conclusion is trivial.
    Otherwise, given an $\eta > 0$, we may consider intervals $I_1, I_2 \subset \sigma(C)$ such that
    \begin{align*}
        I_1 = [\lambda_{\min}, \lambda_{\min}+\eta] \cap \sigma(C), \, I_2 = [\lambda_{\max} - \eta, \lambda_{\max}] \cap \sigma(C)
    \end{align*}
    with the corresponding spectral projections $E(I_1), E(I_2)$.
    Fix two unit vectors, $\ket{\psi_1} \in E(I_1)\cH$ and $\ket{\psi_2} \in E(I_2)\cH$. Direct calculation shows
    \begin{align*}
        \sandwich{\psi_1}{(C-\lambda_{\min}\id_{\cH})}{\psi_1} &= \int_{\sigma(C)}(\lambda - \lambda_{\min}) \, d\mu_1(\lambda) = \int_{I_1}(\lambda - \lambda_{\min}) \, d\mu(\lambda) \leq \eta^2, \\
        \sandwich{\psi_2}{(\lambda_{\max}\id_{\cH} - C)}{\psi_2} &= \int_{\sigma(C)}(\lambda_{\max} - \lambda) \, d\mu_2(\lambda) = \int_{I_2}(\lambda_{\max} - \lambda) \, d\mu(\lambda) \leq \eta^2,
    \end{align*}
    where the measures $\mu_i(X) = \sandwich{\psi_i}{E(X)}{\psi_i}$ are supported on $I_i$ for $i=1, 2$.

    Subsequently, we identify $\cK = \Span{\ket{\psi_1}, \ket{\psi_2}} \in \mathrm{Gr}(2, \cH)$ with projector $P_{\cK}$.
    Clearly $C_{\cK} = P_{\cK} C P_{\cK} \in B(\cK)$ is a two-dimensional self-adjoint operator, so we denote its two eigenvalues by $\mu_{\min}, \mu_{\max}$, and $R(\sigma(C_{\cK})) = 1/2(\mu_{\max}- \mu_{\min})$.
    Then, the above two inequalities show that
    \begin{align*}
        \mu_{\min} \leq \sandwich{\psi_1}{C_{\cK}}{\psi_1} &= \sandwich{\psi_1}{C}{\psi_1} \leq \lambda_{\min} + \eta^2 \\
        \mu_{\max} \geq \sandwich{\psi_2}{C_{\cK}}{\psi_2} &= \sandwich{\psi_2}{C}{\psi_2} \geq \lambda_{\max} - \eta^2.
    \end{align*}
    By Stampfli's Theorem~\ref{thm:StampfliTheorem},
    \begin{align*}
        \inf_{c \in \mathds{C}} \lVert C - c \id_{\cH} \rVert_{\op} = R(\sigma(C)) = \frac{1}{2}(\lambda_{\max} - \lambda_{\min}) \leq \frac{1}{2}(\mu_{\max} - \mu_{\min}) + \eta^2 = \inf_{c \in \mathds{C}} \lVert C_{\cK} - c \id_{\cK} \rVert_{\op} + \eta^2.
    \end{align*}
    
    To upper bound $\inf_{c \in \mathds{C}} \lVert C_{\cK} - c \id_{\cK} \rVert_{\op}$, we work with the eigenbasis $\{ \ket{\mu_{\min}}, \ket{\mu_{\max}} \}$ of $C_{\cK}$ associated with $\{\mu_{\min}, \mu_{\max}\}$.
    In this basis, every unitary $U \in B(\cK)$ satisfies $\lVert U\ket{\mu_{\min}} \rVert_2^2 = \lvert U_{11} \rvert^2 + \lvert U_{21} \rvert^2 = 1$.
    Moreover, if $U$ is also Haar-random, then the two random variables $\lvert U_{11} \rvert^2$ and $\lvert U_{21} \rvert^2$ are identically distributed.
    By symmetry it follows that the expectation values $\mathds{E} \lvert U_{11} \rvert^2 = \mathds{E} \lvert U_{21} \rvert^2 = \frac12$.
    Then, one checks that
    \begin{align*}
        \mathds{E} \lVert [C_{\cK}, U] \rVert_{\op}^2 &\geq \mathds{E} \lVert (C_{\cK}U - UC_{\cK}) \ket{\mu_{\min}} \rVert_{2}^2 \\
        &= \mathds{E} \lVert C_{\cK}U \ket{\mu_{\min}} - U \mu_{\min} \ket{\mu_{\min}} \rVert_{2}^2\\
        &= \mathds{E} \lVert (C_{\cK} - \mu_{\min} \id_{\cK}) U\ket{\mu_{\min}} \rVert_{2}^2 \\
        &= \mathds{E} \left( \lvert U_{11} \rvert^2 (\mu_{\min} - \mu_{\min})^2 + \lvert U_{21} \rvert^2 (\mu_{\max} - \mu_{\min})^2 \right) \\
        &= 2 \frac{1}{4} (\mu_{\max} - \mu_{\min})^2 \\
        &= 2 R(\sigma(C_{\cK}))^2 = 2\inf_{c \in \mathds{C}} \lVert C_{\cK} - c \id_{\cK} \rVert_{\op}^2 \geq 2\inf_{c \in \mathds{C}} \lVert C - c \id_{\cH} \rVert_{\op}^2 - \eta^2.
    \end{align*}
    Since $\eta > 0$ is arbitrary, it follows that
    \begin{equation}\label{eq:ProbabilisticStampfliIntermediateLowerBound}
        \begin{aligned}
            \inf_{c \in \mathds{C}} \lVert C - c \id_{\cH} \rVert_{\op} \leq \inf_{c \in \mathds{C}} \lVert C_{\cK} - c \id_{\cK} \rVert_{\op} = R(\sigma(C_{\cK})), \\
            \text{and} \quad 2\inf_{c \in \mathds{C}} \lVert C - c \id_{\cH} \rVert_{\op}^2 \leq \mathds{E} \lVert [C_{\cK}, U] \rVert_{\op}^2,
        \end{aligned}
    \end{equation}
    using the previous inequality, which will also be used to derive the other upper bound.

    On the other hand, Eq.~\eqref{eq:ProbabilisticStampfliAssumption} is equivalent to $\Pr_{U \sim \mu_{\cK}}\{ \lVert [C, U] \rVert_{\op}^2 \leq \epsilon^2 \} \geq 1 - \delta$, which implies that
    \begin{align*}
        \mathds{E} \lVert [C_{\cK}, U] \rVert_{\op}^2 &= \Pr_{U \sim \mu_{\cK}}\{ \lVert [C_{\cK}, U] \rVert_{\op}^2 \leq \epsilon^2 \} \cdot \mathds{E}( \lVert [C_{\cK}, U] \rVert_{\op}^2 \mid \lVert [C_{\cK}, U] \rVert_{\op}^2 \leq \epsilon^2) \\
        &+ \Pr_{U \sim \mu_{\cK}}\{ \lVert [C_{\cK}, U] \rVert_{\op}^2 > \epsilon^2 \} \cdot \mathds{E}( \lVert [C_{\cK}, U] \rVert_{\op}^2 \mid \lVert [C_{\cK}, U] \rVert_{\op}^2 > \epsilon^2) \\
        &\leq (1-\delta) \epsilon^2 + \delta \cdot \lVert [C_{\cK}, U] \rVert_{\op}^2 \leq (1-\delta) \epsilon^2 + \delta \cdot 4 \lVert C \rVert_{\op}^2.
    \end{align*}
    The first inequality is justified due to reducing the weight of the smaller conditional expectation ($\leq \epsilon^2)$ while increasing the weight of the larger one ($\geq \epsilon^2$) can only enlarge the total, and the second one is a basic calculation.
    It follows from the lower bound Eq.~\eqref{eq:ProbabilisticStampfliIntermediateLowerBound} that
    \begin{align*}
         \inf_{c \in \mathds{C}} \lVert C - c \id_{\cH} \rVert_{\op} \leq \frac{1}{\sqrt{2}} \mathds{E} \lVert [C_{\cK}, U] \rVert_{\op}  \leq  \frac{\sqrt{2}}{2} \left( \sqrt{1-\delta}\, \epsilon + 2\sqrt{\delta}\lVert C \rVert_{\op} \right).
    \end{align*}

    We now derive the upper bound $\lceil{\frac{1}{1-\delta}}\rceil \frac{\epsilon}{2}$ with a Steinhaus-Weil-like argument.
    Let $A_{\cK} = \{U \in U(\cK) \mid \lVert [P_{\cK} C P_{\cK}, U] \rVert_{\op} \leq \epsilon \}$, then Eq.~\eqref{eq:ProbabilisticStampfliAssumption} implies that $\mu_{\cK}(A) \geq 1 -\delta$.
    Also, $A_{\cK}^{-1} = A_{\cK}^* = A_{\cK}$, since for any $U \in A_{\cK}$ one has that
    \begin{align*}
        \lVert [P_{\cK} C P_{\cK}, U^{-1}] \rVert_{\op} = \lVert U [P_{\cK} C P_{\cK}, U^*] U \rVert_{\op} = \lVert [P_{\cK} C P_{\cK}, U] \rVert_{\op}.
    \end{align*}

    Let $l = \lceil{\frac{1}{1-\delta}}\rceil$, we claim that $A_{\cK}^l = U(\cK)$.
    Indeed, since $U(\cK)$ is compact and connected, inductive application of Kemperman's theorem~\cite[Thm.~1.1]{Kemperman1964} implies that
    \begin{align*}
        \mu_{\cK}(A_\cK^l) \geq \min(1, l \cdot \mu_{\cK}(A_\cK)) \geq \lceil{\frac{1}{1-\delta}}\rceil (1-\delta) \geq 1 = \mu_{\cK}(U(\cK)),
    \end{align*}
    consequently $A_{\cK}^l = U(\cK)$.

    Therefore, every $V \in U(\cK)$ is of the form $V = \prod_{i=1}^l U_i$ for some $U_i \in A_{\cK}$.
    By the definition of $A_{\cK}$ and an inductive application of Eq.~\eqref{eq:ProductCommutatorFormula}, we have
    \begin{align}\label{eq:KempermanBoundForAllU}
        \lVert [C_{\cK}, V] \rVert_{\op} = \lVert [C_{\cK}, \prod_{i=1}^l U_i] \rVert_{\op}  \leq l\epsilon = \lceil{\frac{1}{1-\delta}}\rceil \epsilon.
    \end{align}
    But then, by the original Stampfli's Theorem~\ref{thm:StampfliTheorem} and Eq.~\eqref{eq:ProbabilisticStampfliIntermediateLowerBound}, we conclude
    \begin{align*}
        \inf_{c \in \mathds{C}} \lVert C - c \id_{\cH} \rVert_{\op} \leq  \inf_{c \in \mathds{C}} \lVert C_{\cK} - c \id_{\cK} \rVert_{\op} = \frac{1}{2} \sup_{V \in U(\cK)} \lVert [P_{\cK} C P_{\cK}, V] \rVert_{\op} \leq \lceil{\frac{1}{1-\delta}}\rceil \frac{\epsilon}{2}.
    \end{align*}
\end{proof}

\begin{theorem}[Probabilistic Stampfli's theorem]\label{thm:ProbabilisticStampfli}
    Let $\cH$ be a Hilbert space of possibly infinite dimension and let $C \in B(\cH)$.
    Given $\epsilon > 0$ and $\delta \in [0, 1]$, suppose that
    \begin{align}
        \Pr_{U  \sim \mu_{\cK} }\{ \lVert [P_{\cK} C P_{\cK}, U] \rVert_{\op} \leq \epsilon \} \geq 1 - \delta
    \end{align}
    for every subspace $\cK \in \mathrm{Gr}(2, \cH)$ with projector $P_{\cK}$.
    Then
    \begin{align}
        \inf_{c \in \mathds{C}} \lVert C - c \id_{\cH} \rVert_{\op} \leq \min \left( \sqrt{2} \left( \sqrt{1-\delta}\, \epsilon + 2\sqrt{\delta}\lVert C \rVert_{\op} \right),\, \lceil{\frac{1}{1-\delta}}\rceil \frac{\epsilon}{2} \right).
    \end{align}
\end{theorem}
Again, the bound $\lceil{\frac{1}{1-\delta}}\rceil \frac{\epsilon}{2}$ behaves generally better in the high confidence regime ($\delta \to 0$) while the bound $\sqrt{2} \left( \sqrt{1-\delta}\, \epsilon + 2\sqrt{\delta}\lVert C \rVert_{\op} \right)$ is more stable in the low confidence regime ($\delta \to 1$).
\begin{proof}
    To derive the first bound $\sqrt{2} \left( \sqrt{1-\delta}\, \epsilon + 2\sqrt{\delta}\lVert C \rVert_{\op} \right)$, let
    \begin{align*}
        H = \frac{1}{2}(C + C^*),\, K = \frac{1}{2i}(C - C^*)
    \end{align*}
    be the unique self-adjoint operators such that $C = H + iK$.
    For any $\cK \in \mathrm{Gr}(2, \cH)$ with projection $P_{\cK}$, write $C_{\cK} = P_{\cK} C P_{\cK}$, $H_{\cK} = P_{\cK} H P_{\cK}$, and $K_{\cK} = P_{\cK} K P_{\cK}$.
    Clearly both $H_{\cK}$ and $K_{\cK}$ are still self-adjoint.
    Pick $x, y$ as minimizers such that
    \begin{align*}
        \lVert H_{\cK} - x \id_{\cK} \rVert_{\op} &= \inf_{c \in \mathds{C}} \lVert H_{\cK} - c \id_{\cK} \rVert_{\op}, \\
        \lVert K_{\cK} - y \id_{\cK} \rVert_{\op} &= \inf_{c \in \mathds{C}} \lVert K_{\cK} - c \id_{\cK} \rVert_{\op}.
    \end{align*}
    (One can check that $x, y$ are actually the average of the largest and smallest eigenvalues of $H_{\cK}, K_{\cK}$.)
    Note that both $\lVert H_{\cK} \rVert_{\op}, \lVert K_{\cK} \rVert_{\op} \leq \lVert C_{\cK} \rVert_{\op} \leq \lVert C \rVert_{\op}$ by the triangle inequality.
    
    Next, for every unitary $U \in B(\cK)$ direct calculation shows that
    \begin{align*}
        \lVert [C_{\cK}^*, U] \rVert_{\op} = \lVert U [C_{\cK}, U^*] U \rVert_{\op} = \lVert [C_{\cK}, U] \rVert_{\op}.
    \end{align*}
    Then
    \begin{align*}
        \lVert [H_{\cK}, U] \rVert_{\op} \leq \frac{1}{2} (\lVert [C_{\cK}, U] \rVert_{\op} + \lVert [C_{\cK}^*, U] \rVert_{\op}) = \lVert [C_{\cK}, U] \rVert_{\op}
    \end{align*}
    and likewise for $\lVert [K_{\cK}, U] \rVert_{\op} \leq \lVert [C_{\cK}, U] \rVert_{\op}$.
    Therefore, for each Haar-random $U \in B(\cK)$ such that $\lVert [C_{\cK}, U] \rVert_{\op} \leq \epsilon$, the same commutator bounds apply to both $H_{\cK}, K_{\cK}$, i.e.
    \begin{align*}
        \Pr_{U \sim \mu_{\cK}}\{ \lVert [H_{\cK}, U] \rVert_{\op} \leq \epsilon \} = \Pr_{U \sim \mu_{\cK}}\{ \lVert [K_{\cK}, U] \rVert_{\op} \leq \epsilon \} \geq 1 - \delta.
    \end{align*}
    It follows from Lem.~\ref{lem:ProbabilisticStampfliNormal} that both
    \begin{align*}
        \lVert H - x \id_{\cK} \rVert_{\op}, \, \lVert K - y \id_{\cK} \rVert_{\op} \leq \frac{\sqrt{2}}{2} \left( \sqrt{1-\delta}\, \epsilon + 2\sqrt{\delta}\lVert C \rVert_{\op} \right),
    \end{align*}
    which implies
    \begin{align*}
        \inf_{c \in \mathds{C}} \lVert C - c \id \rVert_{\op} & \leq \lVert (H + iK) - (x + iy) \id \rVert_{\op} \\
        & \leq \lVert H - x \id \rVert_{\op} + \lVert i(K - y \id) \rVert_{\op} \leq \sqrt{2} \left( \sqrt{1-\delta}\, \epsilon + 2\sqrt{\delta}\lVert C \rVert_{\op} \right).
    \end{align*}

    Finally, the bound $\lceil{\frac{1}{1-\delta}}\rceil \frac{\epsilon}{2}$ directly follows from Lem.~\ref{lem:ProbabilisticStampfliNormal} since the proof does not rely on the self-adjointness.
\end{proof}

Observe that when the commutator is smaller than $\epsilon$ with high confidence $(\delta < 1/2)$, we recover the original Stampfli's Theorem~\ref{thm:StampfliTheorem} up to a factor of $2$ ($\epsilon$ instead of $\epsilon/2$), which shows up to compensate for the fact that we need to consider $A_{\cK}^2$ (from the proof of Lem.~\ref{lem:ProbabilisticStampfliNormal}).

\begin{remark}\label{rem:CommentOnKisDimofd}
    One can generalize the setting of Thm.~\ref{thm:ProbabilisticStampfli} to Haar-random unitaries $U \in B(\cK)$ when $2 \leq \dim(\cK) < \infty$, at the cost of having a slightly worse constant factor:
    \begin{align}\label{eq:KDimofRemark}
        \inf_{c \in \mathds{C}} \lVert C - c \id_{\cH} \rVert_{\op} \leq \min \left( 2\sqrt{2} \left( \sqrt{1-\delta}\, \epsilon + 2\sqrt{\delta}\lVert C \rVert_{\op} \right),\, \lceil{\frac{1}{1-\delta}}\rceil \frac{\epsilon}{2} \right).
    \end{align}
    The second bound $\lceil{\frac{1}{1-\delta}}\rceil \frac{\epsilon}{2}$ can be shown the same way as that of Lem.~\ref{lem:ProbabilisticStampfliNormal} as the argument only relies on the compactness and connectedness of finite-dimensional unitary groups.
    We now give a sketch of the proof for the first bound.

    For simplicity assume $\cK=\cH$ and $C=C^* \in B(\cK)$.
    Let $\lambda_{\max}$ be the maximal eigenvalue of $C$ with eigenvector $\ket{\lambda_{\max}}$ and $\lambda_{\min}$ be the minimal eigenvalue with eigenvector $\ket{\lambda_{\min}}$, and let $R = (\lambda_{\max} - \lambda_{\min})/2$.
    By the pigeonhole principle, at least $\lceil{d/2}\rceil$ of the eigenvalues lie in the interval $[\lambda_{\max} - R, \lambda_{\max}]$ or in $[\lambda_{\min}, \lambda_{\min}+R]$.
    Without loss of generality we assume the former so that there are $\geq d/2$ of them are in $[\lambda_{\max} - R, \lambda_{\max}]$.

    In the eigenbasis $\{\ket{\lambda_{\min}}, \dots, \ket{\lambda_{\max}}\}$ of $C$, the vector $U\ket{\lambda_{\min}}$ is the first column of a Haar-random unitary $U \in U(d)$.
    Hence by the same symmetry argument that $\mathds{E} \rvert U_{i1} \rvert^2 = 1/d$ for each $i$.
    A direct calculation shows
    \begin{align*}
        \mathds{E} \lVert[C, U] \rVert_{\op}^2 &\geq \mathds{E} \lVert (C - \lambda_{\min} \id) U\ket{\lambda_{\min}} \rVert_{2}^2 = \mathds{E}\sum_i \lvert \lambda_i - \lambda_{\min} \rvert^2 \lvert U_{ia} \rvert^2 \\
        &\geq \mathds{E}\sum_{\lambda_i \in [\lambda_{\max} - R, \lambda_{\max}]} \lvert \lambda_i - \lambda_{\min} \rvert^2 \lvert U_{ia} \rvert^2 \geq \frac{d}{2} R^2 \frac{1}{d} = \frac{R^2}{2}.
    \end{align*}
    The exact same proof then leads to $\sqrt{2}$ factor for the self-adjoint case and consequently $2\sqrt{2}$ for the general case.
\end{remark}

\subsection{Doubly probabilistic Stampfli's theorem}\label{sec:DoublyProbabilistic}
While Thm.~\ref{thm:ProbabilisticStampfli} is a proper generalization of the original Stampfli's theorem, we note that the Haar-random single-qubit assumption is still not physical enough.
Indeed, it requires verifications of almost commutation over all single-qubit subspaces $\cK \in \mathrm{Gr}(2, \cH)$, which is unrealistic.

We therefore consider a doubly probabilistic generalization: also randomly sample two-dimensional subspaces $\cK \in \mathrm{Gr}(2, \cH)$ and then check the almost commutation for Haar-random unitaries in $B(\cK)$.
This is far more reasonable in physical implementations.

However, the random sampling of two-dimensional subspaces in infinite-dimensional $\cH$ does not make sense.
In fact, it is well-known that the Grassmannian $\mathrm{Gr}(2, \cH)$ does not admit a non-trivial, $\sigma$-finite, $U(\cH)$-invariant Borel measure when $\dim(\cH) = \infty$.
Hence, we consider the finite-dimensional setting for the doubly probabilistic generalization.

Now, $\mathrm{Gr}(2, \cH)$ does admit a probability measure $\nu_{\mathrm{Gr}(2, \cH)}$ when $\dim(\cH) = d < \infty$.
In particular, the notion of a random two-dimensional subspace $\cK$ is equivalent to the following:
\begin{enumerate}[label=(\alph*)]
    \item Fix two orthonormal vectors $\ket{v_1}, \ket{v_2} \in \cH$ as the reference two-dimensional subspace.
    \item There exists some Haar-random $U(d)$-unitary $V \in B(\cH)$ such that $\cK = V \Span{\ket{v_1}, \ket{v_2}}$.
\end{enumerate}
Thanks to the invariance of Haar measures, $\ket{v_1}$ and $\ket{v_2}$ can be chosen arbitrarily.
This allows us to formulate and prove another generalization of Stampfli's theorem.

\begin{theorem}[Doubly probabilistic Stampfli's theorem]\label{thm:DoublyProbabilisticStampfli}
    Let $\cH$ be a $d$-dimensional Hilbert space and let $C \in B(\cH)$.
    Given $\epsilon > 0$ and $\delta, \eta \in [0,1]$, suppose that
    \begin{align}\label{eq:DoublyProbabilisticAssumption}
        \Pr_{\cK \sim \nu_{\mathrm{Gr}(2, \cH)}} \left\{ \Pr_{U  \sim \mu_{\cK} } \{ \lVert [P_{\cK} C P_{\cK}, U] \rVert_{\op} \leq \epsilon \} \geq 1 - \delta \right\} \geq 1 - \eta,
    \end{align}
    where $P_{\cK}$ denotes the projection onto $\cK$.
    Then
    \begin{equation}\label{eq:DoublyProbabilisticResult}
        \begin{aligned}
            \inf_{c\in\mathds C}\bigl\lVert C - c\,\id_{\cH}\bigr\rVert_{\op}
            \le 2\sqrt{\frac{d^2-1}{6}}
            \min\Biggl(
            &\sqrt{(1-\eta)(1-\delta)}\,\epsilon
            + 2\lVert C \rVert_{\op}\sqrt{\delta(1-\eta)+\eta},\\
            &\sqrt{1-\eta}\,\Bigl\lceil \frac{1}{1-\delta} \Bigr\rceil \epsilon
            + 2\lVert C \rVert_{\op}\sqrt{\eta}
            \Biggr).
        \end{aligned}
    \end{equation}
    and this upper bound necessarily depends on the dimension $d$.
    In addition, the leading factor can be reduced to $\sqrt{(d^2-1)/6}$ when $C$ is self-adjoint.
\end{theorem}
Depending on the specific values of $(\epsilon, \delta, \eta, \lVert C \rVert_{\op})$, one bound might be tighter than the other.
\begin{proof}
    Let us begin with the first bound $2\sqrt{\frac{d^2-1}{6}}\left( \sqrt{(1-\eta)(1-\delta)}\, \epsilon + 2\lVert C \rVert_{\op} \sqrt{\delta(1-\eta) + \eta} \right)$.
    It is sufficient to show the case when $C$ is self-adjoint, as the general case follows by the same argument in the proof of Thm.~\ref{thm:ProbabilisticStampfli}.
    Analogous to Lem.~\ref{lem:ProbabilisticStampfliNormal}, here we instead try to bound $\mathds{E}_{\cK \sim \nu_{\mathrm{Gr}(2, \cH)}} \mathds{E}_{U  \sim \mu_{\cK} } \lVert [C_{\cK}, U] \rVert_{\op}$ for $C_{\cK} = P_{\cK} C P_{\cK}$.

    By the definition of expectation values and the trivial commutator bound, the upper bound is straightforward: 
    \begin{align}\label{eq:DoublyProbabilisticStampfliIntermediateUpperBound}
        \mathds{E}_{\cK \sim \nu_{\mathrm{Gr}(2, \cH)}} \mathds{E}_{U  \sim \mu_{\cK} } \lVert [C_{\cK}, U] \rVert_{\op}^2 \leq (1-\eta) \cdot \left( (1-\delta)\epsilon^2 + 4\delta \lVert C \rVert_{\op}^2 \right) + \eta \cdot 4\lVert C \rVert_{\op}^2.
    \end{align}

    For the lower bound, Eq.~\eqref{eq:ProbabilisticStampfliIntermediateLowerBound} already shows that
    \begin{align*}
        \mathds{E}_{U  \sim \mu_{\cK} } \lVert [C_{\cK}, U] \rVert_{\op}^2 \geq 2 R(\sigma(C_{\cK}))^2,
    \end{align*}
    where $R(\sigma(C_{\cK}))$ is the radius of the spectrum of $C_{\cK}$.
    Thus, the rest of the proof amounts to calculating $\mathds{E}_{\cK \sim \nu_{\mathrm{Gr}(2, \cH)}} R(\sigma(C_{\cK}))$.

    We first calculate $R(\sigma(C_{\cK}))^2$.
    To this end, denote by $\lambda_i$ the eigenvalues of $C$ with the corresponding eigenvectors $\ket{\lambda_i}$.
    By the discussion preceding the theorem, there exists some Haar-random $U(d)$-unitary $V \in B(\cH)$ such that $\cK = V \Span{\ket{\lambda_1}, \ket{\lambda_2}}$.
    It follows from $C = \sum_k \lambda_k \ketbra{\lambda_k}{\lambda_k}$ that
    \begin{align*}
        ( C_{\cK} )_{ij} = \sandwich{\lambda_i}{V^* C V}{\lambda_j} = \sum_k \lambda_k \Bar{V}_{ki} V_{kj}
    \end{align*}
    in the basis $\{ V\ket{\lambda_1}, V\ket{\lambda_2} \}$.
    Then
    \begin{align*}
        R(\sigma(C_{\cK}))^2 &= \frac{1}{4}( \Tr{C_{\cK}}^2 - 4 \det(C_{\cK})) \\
        &= \sum_{k, l} \lambda_k \lambda_l \left( \Bar{V}_{k1} V_{k1} \Bar{V}_{l1} V_{l1} + \Bar{V}_{k2} V_{k2} \Bar{V}_{k2} V_{k2} - 2 \Bar{V}_{k1} V_{k1} \Bar{V}_{l2} V_{l2} + 4 \Bar{V}_{k1} V_{k2} \Bar{V}_{l2} V_{l1} \right),
    \end{align*}
    using the convenient trace-determinant formula for $2 \times 2$ matrices.

    To compute $\mathds{E}_{\cK \sim \nu_{\mathrm{Gr}(2, \cH)}} (V_{k_1 a_1} V_{k_2 a_2} \Bar{V}_{l_1 b_1} \Bar{V}_{l_2 b_2})$ for a Haar-random unitary $V \in B(\cH)$, we can use Weingarten calculus~\cite[Cor.~2.4]{collins2006integration}.
    One checks that
    \begin{align*}
        \mathds{E}_{\cK \sim \nu_{\mathrm{Gr}(2, \cH)}} R(\sigma(C_{\cK}))^2 &= \frac{1}{4} \sum_{k, l} \lambda_k \lambda_l \frac{6 \delta_{kl} - 6}{d(d^2-1)} \\
        &= \frac{3}{2d(d^2-1)} \left( d \Tr{C^2} - \Tr{C}^2 \right) \geq \frac{3}{(d^2-1)} R(\sigma(C))^2.
    \end{align*}
    The last inequality with factor $(d^2-1)$ is in fact sharp.
    To see this, observe that $(d \Tr{C^2} - \Tr{C}^2)/d^2$ is the variance of $\{\lambda_i\}$ with uniform distribution.
    Given a fixed $R(\sigma(C)) = (\lambda_{\max} - \lambda_{\min})/2$, the variance is minimized when all non-extremal eigenvalues $\lambda_i = (\lambda_{\max} + \lambda_{\min})/2$, whence $d \Tr{C^2} - \Tr{C}^2 = 2dR(\sigma(C))$.
    Another sanity check is to consider $d=2$ which gives the trivial constant $1$.
    
    It follows that
    \begin{align*}
        \mathds{E}_{\cK \sim \nu_{\mathrm{Gr}(2, \cH)}} \mathds{E}_{U  \sim \mu_{\cK} } \lVert [C_{\cK}, U] \rVert_{\op}^2 \geq \frac{3 R(\sigma(C))^2}{(d^2-1)},
    \end{align*}
    and we are done by Eq.~\eqref{eq:DoublyProbabilisticStampfliIntermediateUpperBound}.

    Finally, we show the second bound $\sqrt{1-\eta}\,\Bigl\lceil \frac{1}{1-\delta} \Bigr\rceil \epsilon
            + 2\lVert C \rVert_{\op}\sqrt{\eta}$.
    Recall Eq.~\eqref{eq:KempermanBoundForAllU} from the proof of Lem.~\ref{lem:ProbabilisticStampfliNormal}, every $V \in U(\cK)$ satisfies that $\lVert [C_{\cK}, V] \rVert_{\op} \leq \lceil{\frac{1}{1-\delta}}\rceil \epsilon$.
    That is, 
    \begin{align*}
        \Pr_{\cK \sim \nu_{\mathrm{Gr}(2, \cH)}} \left\{ \Pr_{V  \sim \mu_{\cK} } \{ \lVert [P_{\cK} C P_{\cK}, V] \rVert_{\op} \leq \lceil{\frac{1}{1-\delta}}\rceil\epsilon \} \geq 1 - 0 \right\} \geq 1 - \eta.
    \end{align*}
    Then we are done by setting $\epsilon \to \lceil{\frac{1}{1-\delta}}\rceil\epsilon$ and $\delta \to 0$ into the first bound.
\end{proof}

Observe that if one has the power to verify the $\epsilon$-almost commutation directly over the unitary group $U(\cH)$ of $\cH$ beyond just the two-dimensional qubit unitaries, then Rem.~\ref{rem:CommentOnKisDimofd} applies and the conclusion of Thm.~\ref{thm:DoublyProbabilisticStampfli} simplifies to Eq.~\eqref{eq:KDimofRemark}.

\subsection{Approximate Tsirelson's theorem from Haar-random single-qubit unitary}
We refer to the method of randomly sampling a single-qubit subspace $\cK \in \mathrm{Gr}(2, \cH)$ and then certifying commutation with Haar-random $U(2)$-unitaries in $B(\cK)$ as \emph{Haar-random single-qubit unitary} sampling.
Thus, with the doubly probabilistic Stampfli's theorem as in Thm.~\ref{thm:DoublyProbabilisticStampfli}, we can analogously formulate and prove an approximate version of Tsirelson's theorem, resulting in another version of the quantitative Ozawa's result~\cite{ozawa2013tsirelson}.

\begin{lemma}\label{lem:BipartitionProbStampfli}
    Consider two Hilbert spaces $\cH_1$ with dimension $d_1$ and $\cH_2$ with dimension $d_2$, and let $C \in B(\cH_1 \otimes \cH_2)$.
    Given $\epsilon > 0$ and $\delta, \eta \in [0,1]$, suppose that
    \begin{align}
        \Pr_{\cK \sim \mu_{\mathrm{Gr}(2, \cH_2)}} \left\{ \Pr_{U  \sim \mu_{\cK} } \{ \lVert [ (\id_1 \otimes P_{\cK}) C (\id_1 \otimes P_{\cK}), \id_1 \otimes U] \rVert_{\op} \leq \epsilon \} \geq 1 - \delta \right\} \geq 1 - \eta,
    \end{align}
    where $P_{\cK}$ denotes the projection onto $\cK \subset \cH_2$.
    
    Then $C' = 1/d_2\tr{\cH_2}{C} \in B(\cH_1)$ satisfies
    \begin{equation}
        \begin{aligned}
            \Vert C - C' \otimes \id_2 \rVert_{\op}
            \leq 2d_1\sqrt{\frac{d_2^2-1}{6}}
            \min\Biggl(
            &\sqrt{(1-\eta)(1-\delta)}\,\epsilon
            + 2\lVert C \rVert_{\op}\sqrt{\delta(1-\eta)+\eta},\\
            &\sqrt{1-\eta}\,\Bigl\lceil \frac{1}{1-\delta} \Bigr\rceil \epsilon
            + 2\lVert C \rVert_{\op}\sqrt{\eta}
            \Biggr).
        \end{aligned}
    \end{equation}
    Consequently, if $C$ is positive semidefinite then so is $C'$.
    Note the leading constant $2d_1$ can be improved to $d_1$ when $C$ is self-adjoint.
\end{lemma}
\begin{proof}
    This proof is almost identical to that of Lem.~\ref{lem:BipartitionApproxSchur}.
    Adopting the same notation, we point out the only difference: Thm.~\ref{thm:DoublyProbabilisticStampfli} implies that
    \begin{equation*}
        \begin{aligned}
            \lVert C_{(kl)} - c_{kl}\id_2 \rVert_{\op} \leq 2\sqrt{\frac{d_2^2-1}{6}}
            \min\Biggl(
            &\sqrt{(1-\eta)(1-\delta)}\,\epsilon
            + 2\lVert C \rVert_{\op}\sqrt{\delta(1-\eta)+\eta},\\
            &\sqrt{1-\eta}\,\Bigl\lceil \frac{1}{1-\delta} \Bigr\rceil \epsilon
            + 2\lVert C \rVert_{\op}\sqrt{\eta}
            \Biggr),
        \end{aligned}
    \end{equation*}
    where we use the fact that $\lVert C \rVert_{\op} \geq \lVert C_{(kl)} \rVert_{\op}$.
    Hence,
    \begin{align*}
        \Vert C - C' \otimes \id_2 \rVert_{\op} \leq d_1 \max_{k,l} \lVert C_{(kl)} - c_{kl}\id_2 \rVert_{\op}
    \end{align*}
    gives the desired bound.
    Note that the above inequality is sharp when all $c_{kl}$ are the same, meaning that the dimension scaling $d_1$ is also unavoidable.
\end{proof}

For the non-simple version of approximate Tsirelson's theorem, similarly to Thm.~\ref{thm:ApproximateTsirelsonGeneral}, one needs to assume additional commutator bounds with the simple projections.
\begin{theorem}[Approximate Tsirelson's theorem, Haar-random unitary case]\label{thm:ApproximateTsirelsonProbabilistic}
    Let $\cA$ be generated by finitely-many contractive self-adjoint operators $\{ \A \} \subset B(\cH)$ and $\cB$ be generated by finitely-many contractive self-adjoint operators $\{ \B \} \subset B(\cH)$ for some $d$-dimensional Hilbert space $\cH$.
    Suppose that $\cA$ admits the Artin-Wedderburn decomposition
        \begin{align*}
            \cA = \bigoplus_{l=1}^L \cA_l \simeq \bigoplus_{l=1}^L B(\cH_A^l) \otimes \id_B^l \text{ and } \A = \bigoplus_{l=1}^L \A^l \otimes \id_B^l,
        \end{align*}
    with the corresponding orthogonal projectors $\Pi_l$ onto the direct summands.
    
    Given $\epsilon > 0$ and $\delta, \eta \in [0,1]$.
    Suppose that,
    \begin{align}\label{eq:HaarRandomAssumptionBipartite}
        \Pr_{\cK^l \sim \nu_{\mathrm{Gr}(2, \cH_A^l)}} \left\{ \Pr_{U \sim \mu_{\cK^l}} \{ \lVert [[U_l \otimes \id_B^l, (P_{\cK^l} \otimes \id_B^l) \Pi_l \B \Pi_l (P_{\cK^l} \otimes \id_B^l)] \rVert_{\op} \leq \epsilon \} \geq 1 - \delta \right\} \geq 1 - \eta,
    \end{align}
    where $P_{\cK^l}$ denotes the projection onto $\cK^l \subset \cH_A^l$.
    Furthermore, assume that for all $b, y, l$
    \begin{align}\label{eq:SimpleProjectionCommuteAssumption}
        \lVert [\Pi_l, \B] \rVert_{\op} \leq \epsilon.
    \end{align}
    
    Then there exist operators $\B' \in \bigoplus_{l=1}^L \id_A^l \otimes B(\cH_B^l) = \cA'$ such that, for all $b, y$,
    \begin{equation}\label{eq:ApproximateTsirelsonBoundProbabilistic}
        \begin{aligned}
            \lVert \B - \B' \rVert_{\op}
            \leq d(d-1)\epsilon +  d\sqrt{\frac{d^2-1}{6}}
            \min\Biggl(
            &\sqrt{(1-\eta)(1-\delta)}\,\epsilon
            + 2\sqrt{\delta(1-\eta)+\eta},\\
            &\sqrt{1-\eta}\,\Bigl\lceil \frac{1}{1-\delta} \Bigr\rceil \epsilon
            + 2\sqrt{\eta}
            \Biggr).
        \end{aligned}
    \end{equation}
    In addition, if $\B$ is positive then so is $\B'$.
    Note that the bound has an $O(d^2\epsilon)$ scaling.
\end{theorem}
\begin{proof}
    The proof is analogous to that of Thm.~\ref{thm:ApproximateTsirelsonGeneral} so we only give a sketch.
    Since $\B$ is self-adjoint and $\lVert \B \rVert_{\op} \leq 1$, by Lem.~\ref{lem:BipartitionProbStampfli}, there exists for each $l$ some positive $\B^l \in B(\cH_B^l)$ such that
    \begin{equation*}
        \begin{aligned}
            \lVert \Pi_l \B \Pi_l - \id_A^l \otimes \B^l \rVert_{\op} \leq d_l\sqrt{\frac{d_l^2-1}{6}}
            \min\Biggl(
            &\sqrt{(1-\eta)(1-\delta)}\,\epsilon
            + 2\lVert C \rVert_{\op}\sqrt{\delta(1-\eta)+\eta},\\
            &\sqrt{1-\eta}\,\Bigl\lceil \frac{1}{1-\delta} \Bigr\rceil \epsilon
            + 2\lVert C \rVert_{\op}\sqrt{\eta}
            \Biggr),
        \end{aligned}
    \end{equation*}
    where $d_l = \dim(\cH_A^l \otimes \cH_B^l)$.
    The commutation assumption of $\B$ with $\Pi_l$ implies
    \begin{align*}
         \lVert \B - \sum_l \Pi_l \B \Pi_l \rVert_{\op} \leq L(L-1) \epsilon \leq d(d-1) \epsilon.
    \end{align*}
    Define $\B' := \bigoplus_{l=1}^L \id_A^l \otimes \B^l$, and we are done by the triangle inequality.
\end{proof}

\section{Applications and outlook}\label{sec:ApplicationOutlook}
Our main results, Thm.~\ref{thm:ApproximateTsirelsonGeneral} and Thm.~\ref{thm:ApproximateTsirelsonProbabilistic}, establish quantitative approximate versions of Tsirelson's theorem from two distinct perspectives on almost commutation.
The first approach (Thm.~\ref{thm:ApproximateTsirelsonGeneral}) provides an error guarantee contingent on potentially hard-to-determine algebraic complexity parameters $c_1, c_2, c_3$ (discussed in Rem.~\ref{rem:ScalingFactors}), scaling roughly as $O(c_1^2 c_2^2 c_3^2 d^2 \epsilon)$.
The second, doubly probabilistic formulation (Thm.~\ref{thm:ApproximateTsirelsonProbabilistic}) yields a bound scaling as $O(d^2 \epsilon)$ with prefactors dependent on probabilistic confidence parameters ($\delta, \eta$) of random unitary sampling and assumptions about commutation with simple-block projectors.
This offers a trade-off, making the latter potentially advantageous when algebra generation is difficult, but probabilistic checks are feasible.
Both methods confirm the overall $O(d^2\epsilon)$ error bounds, and these discussions can be generalized to multipartite scenarios by induction, akin to the original Tsirelson's theorem.

In this last section, we begin with an application to quantum information theory, as teased in Sec.~\ref{sec:QITmotivation}: the construction of an approximating tensor-product strategy (Prop.~\ref{prop:TensorStrategyConstruction}) based on our main theorems.
Next, we discuss the implications of our results in the context of the NPA hierarchy and the known computational complexity results, giving scenarios with non-negligible approximation error (Rem.~\ref{rem:ConnectionToNPAComplexity}).
Then, we connect our works to the broader historical context on approximating almost commuting matrices with genuinely commuting ones (Sec.~\ref{sec:RelateToPreviousAlmostCommuting}).
Lastly, we finish by discussing the implications of our work and outlining future possible research directions.

\subsection{Constructing tensor-product approximation}
For clarity, a \emph{tensor-product quantum strategy} is a triple $(\{A_{a|x}\},\{B_{b|y}\},\rho)$ on a Hilbert space $\cH = \cH_A \otimes \cH_B$, where for each $x$ the family $\{A_{a|x}\} \subset B(\cH_A)$ is a POVM (i.e., $\A \geq 0$ and $\sum_a \A = \id_A$), for each $y$ the family $\{B_{b|y}\} \subset B(\cH_B)$ is a POVM, and $\rho \in B(\cH)$ is a positive operator with $\Tr{\rho}=1$.
The associated \emph{correlations} are
\begin{align*}
    p(ab|xy)=\Tr{\rho \cdot (A_{a|x} \otimes B_{b|y})}.
\end{align*}
By contrast, a \emph{commuting operator quantum strategy} is a triple $(\{A_{a|x}\},\{B_{b|y}\},\rho)$ on a single Hilbert space $\cH$, where $\{A_{a|x}\}, \{B_{b|y}\} \subset B(\cH)$ are POVMs such that $[\A, \B] = 0$ for all $a, b, x, y$ with correlations
\begin{align*}
    p(ab|xy) = \Tr{ \rho \cdot \A \cdot \B}.
\end{align*}
As recalled in Sec.~\ref{sec:QITmotivation}, these two models are equivalent when $\cH$ is finite-dimensional by Tsirelson's theorem.

In what follows, we relax exact commutation to \emph{$\epsilon$-almost} commutation (in the senses made precise in Theorems~\ref{thm:ApproximateTsirelsonGeneral} and~\ref{thm:ApproximateTsirelsonProbabilistic}) and show how our results yield a \emph{constructive} procedure: from any $d$-dimensional $\epsilon$-almost commuting strategy, one can build a tensor-product strategy whose correlations approximate the original ones up to $O(d^2\epsilon)$.

\begin{proposition}\label{prop:TensorStrategyConstruction}
    Given a quantum strategy $(\A, \B, \rho)$ on a $d$-dimensional Hilbert space $\cH$ that is $\epsilon$-almost commuting (in the sense of Thms.~\ref{thm:ApproximateTsirelsonGeneral} or~\ref{thm:ApproximateTsirelsonProbabilistic}).
    Then there exist Hilbert spaces $\cH_A, \cH_B$ and a tensor-product strategy $(\Tilde{A}_{a|x}, \Tilde{B}_{b|y}, \Tilde{\rho})$ on $\cH_A \otimes \cH_B$ such that its correlations are $O(d^2\epsilon)$-close to those of the original $\epsilon$-almost commuting strategy $(\A, \B, \rho)$.
\end{proposition}
\begin{proof}
    The Artin-Wedderburn decomposition of the algebra generated by $\{\A\}$ implies that $\cH = \bigoplus_l \cH_A^l \otimes \cH_B^l$ with
    \begin{align*}
        \A = \bigoplus_l \A^l \otimes \id_B^l \in \bigoplus_l B(\cH_A^l) \otimes \id_B^l.
    \end{align*}
    Let $\Pi_l$ be the orthogonal projectors to the $i$-th summand $B(\cH_A^l) \otimes \id_B^l$.
    By Thm.~\ref{thm:ApproximateTsirelsonGeneral} or~\ref{thm:ApproximateTsirelsonProbabilistic}, for all $b, y$, the operators $\B$ can be approximated within $O(d^2\epsilon)$ in operator norm by positive operators
    \begin{align*}
        \B' = \bigoplus_l \id_A^l \otimes \frac{1}{d_A^l}\tr{\cH_A^l}{\Pi_l \B \Pi_l} \in \bigoplus_l \id_A^l \otimes B(\cH_B^l).
    \end{align*}
    They satisfy $\sum_b \B' = \id$ since the normalized partial trace is unital and $\sum_b \B = \id$, hence forming POVMs.
    This yields an exactly commuting strategy $(\A, \B', \rho)$ on $\cH$ which is $O(d^2\epsilon)$-close to the original strategy.
    
    We construct the equivalent tensor-product strategy $(\Tilde{A}_{a|x}, \Tilde{B}_{b|y}, \Tilde{\rho})$ as follows:
    Let $\cH_A = \bigoplus_l \cH_A^l$ and $ \cH_B = \bigoplus_l \cH_B^l$.
    Define new operators
    \begin{align*}
        \Tilde{A}_{a|x} = \bigoplus_l \A^l \in B(\cH_A),\, \Tilde{B}_{b|y} = \bigoplus_l\frac{1}{d_A^l}\tr{\cH_A^l}{\Pi_l \B \Pi_l} \in B(\cH_B),
    \end{align*}
    and the new state
    \begin{align*}
        \Tilde{\rho} = \iota(\rho) \in B(\cH_A \otimes \cH_B)
    \end{align*}
    via the natural embedding $\iota: B(\bigoplus_l \cH_A^l \otimes \cH_B^l) \to B(\cH_A \otimes \cH_B)$.
    Thanks to the block structure of $\Tilde{\rho}$, one can directly check that the correlations of $(\A, \B', \rho)$ are preserved by $(\Tilde{A}_{a|x}, \Tilde{B}_{b|y}, \Tilde{\rho})$, consequently $O(d^2\epsilon)$-close to those of the original strategy.
\end{proof}
We note that our constructions recover the standard Tsirelson's theorem asymptotically as $\epsilon \to 0$, indicating our result is a quantitative version of~\cite{ozawa2013tsirelson}.

\subsection{NPA hierarchy and computational complexity}
We comment on an interesting consequence of our approximate Tsirelson's theorem in relation to the NPA hierarchy~\cite{navascues2008convergent, pironio2010convergent}, an important tool in the studies of quantum correlations.
This connection has implications for understanding when the approximation error from our theorem must necessarily be significant.

\begin{remark}\label{rem:ConnectionToNPAComplexity}
    The NPA hierarchy provides a sequence of constraints, indexed by level $N$, that characterize correlations arising from commuting observable strategies.
    This hierarchy is complete in the limit $N \to \infty$.
    At a finite level $N$, an NPA strategy $S_N$ can be realized in a $d_N$-dimensional Hilbert space and involves observables that are $O(1/\sqrt{N})$-almost commuting~\cite[Thm.~23]{coudron2015interactive}.
    
    Our Thm.~\ref{thm:ApproximateTsirelsonGeneral} states that such an $O(1/\sqrt{N})$-almost commuting strategy $S_N$ can be approximated by a genuine tensor-product strategy with an operator norm error of $O(d_N^2/\sqrt{N})$.
    We argue that this error term cannot always vanish as $N \to \infty$ due to computational complexity arguments.
    \begin{enumerate}
        \item Consider the result $\mathrm{MIP}^* = \mathrm{RE}$~\cite{ji2021mip}.
        This implies there are problems (specifically, $\mathrm{RE}$-hard problems) for which the closure of the set of correlations achievable with tensor-product strategies ($C_{qa}$) is strictly smaller than the set achievable with commuting observable strategies ($C_{qc}$), i.e., $C_{qa} \subsetneq C_{qc}$. 
        The NPA strategies $S_N$ generate correlations that converge towards $C_{qc}$.
        Our approximation, being a tensor-product strategy, generates correlations within $C_{qa}$.
        
        Hence, the approximation error $O(d_N^2/\sqrt{N})$ must be generally \emph{non-vanishing} in the limit $N \to \infty$.
        If not, i.e., the error vanished, it would imply $C_{qa}$ could approximate $C_{qc}$ arbitrarily well, contradicting the known set separation.

        \item A similar line of reasoning applies to the conjecture $\mathrm{MIP}^\mathrm{co} = \mathrm{coRE}$~\cite{ji2021mip} (more precisely, the gaped decision problem of quantum commuting value is $\mathrm{coRE}$-hard).
        If this conjecture holds, it would imply the existence of $\mathrm{coRE}$-hard problems where $O(1/\sqrt{N})$-almost commuting strategies $S_N$ can achieve outcomes (e.g., Bell scores) significantly larger than those achievable by any strictly commuting observable strategy (and thus, by any tensor-product strategy).
        In such a scenario, these $S_N$ strategies would be inherently ``far'' from any tensor-product approximation.
        Consequently, our approximation of $S_N$ by a tensor-product strategy must necessarily result in a non-vanishing error of $O(d_N^2/\sqrt{N})$ to account for this performance gap.
    \end{enumerate}

    In essence, these complexity results highlight scenarios where the distinction between almost-commuting and strictly commuting (or tensor-product) models is presented.
    Our quantitative theorems provide a bound on how well one can bridge this distinction, and these complexity results suggest that our error bound, or indeed any such bound, cannot universally tend to zero.
\end{remark}

\subsection{Relation to prior works on almost commuting matrices}\label{sec:RelateToPreviousAlmostCommuting}
The question of whether matrices or operators that almost commute are necessarily close to a genuinely commuting pair is a longstanding problem with a rich history, initiated by Rosenthal~\cite{rosenthal1969almost} for the normalized Hilbert-Schmidt norm and Halmos~\cite{halmos1976some} for the operator norm.
Early studies in the operator norm (e.g., \cite{luxemburg1970almost,pearcy1979almost}) often yielded affirmative answers, though typically with dimension-dependent error bounds.

The search for dimension-independent bounds revealed a crucial dichotomy in the operator norm.
Voiculescu~\cite{voiculescu1983asymptotically} showed that almost commuting unitary matrices need not be close to commuting ones (incidentally by considering clock-and-shift matrices $\Sigma_3, \Sigma_1$), and Choi~\cite{choi1988almost} extended this negative result to general matrices.
In contrast, Lin's theorem~\cite{lin1996almost,friis1996almost} provided a positive dimension-independent answer for a pair of self-adjoint matrices.
Recently, extending the scope to infinite dimensions and multiple operators,~\cite{lin2024almost} has connected the approximability of self-adjoint operators to spectral properties.

In parallel, the searches for dimension-independent bounds in the normalized Hilbert-Schmidt norm established affirmative results for a pair of normal~\cite{glebsky2010almost}, self-adjoint~\cite{filonov2010hilbert}, and unitary matrices~\cite{hadwin2018stability}.
More recently, Ioana~\cite{ioana2024almost} further confirms approximability if at least one matrix is normal, while showing a negative result for general matrices.

Our work contributes to this area by considering two (thus, inductively, multiple) finite-dimensional $C^*$-algebras whose generators almost commute, formulated either in terms of operator norm bounds against specific matrix generators (like clock-and-shift matrices) or via a probabilistic formulation involving Haar-random unitaries.
For both, we characterize how close these algebras are to having genuinely commuting counterparts (or admitting an approximate tensor-product structure by Prop.~\ref{prop:TensorStrategyConstruction}) in the operator norm, deriving bounds that exhibit dependence on the dimension $d$.
Given the discussion in Rem~\ref{rem:ConnectionToNPAComplexity} based on known separations due to computational complexity results~\cite{ji2021mip}, we do not expect such dimension dependence in the error bounds to be removable.

\subsection{Discussions and future directions}\label{sec:FinalDiscussion}
Fundamentally, Tsirelson's theorem connects the tensor-product formalism with the commuting observable formalism for composite finite-dimensional quantum systems.
While strict commutation can be conceptually enforced by space-like separation, many physical scenarios or experimental setups might only guarantee approximate independence due to correlated noise, imperfect isolation, or other constraints.
Our approximate Tsirelson's theorems (Thm.~\ref{thm:ApproximateTsirelsonGeneral} and Thm.~\ref{thm:ApproximateTsirelsonProbabilistic}) show that Tsirelson's conclusion is robust to such imperfections.
They guarantee that $\epsilon$-almost commuting observables (in either the deterministic or probabilistic sense) necessarily imply that the system's correlations are $O(d^2\epsilon)$-close in operator norm to those of genuine tensor-product quantum correlations.
This validates the use of tensor-product formulation as an effective model even when subsystem independence is only approximately satisfied.

As detailed above in Rem.~\ref{rem:ConnectionToNPAComplexity}, our findings interface with the NPA hierarchy.
This connection is crucial, as it highlights, through computational complexity results like $\mathrm{MIP}^* = \mathrm{RE}$~\cite{ji2021mip}, that for certain problems the approximation error $O(d_N^2/\sqrt{N})$ from our theorems cannot be universally negligible.
This signifies a fundamental limitation in approximating certain almost-commuting strategies with tensor-product strategies, a limitation our quantitative error bounds necessarily reflect.
Conversely, for scenarios without this intrinsic separation, improving our error bounds remains interesting for applications like robust self-testing~\cite{vsupic2020self}.

Furthermore, our probabilistic Stampfli's theorems (Thm.~\ref{thm:ProbabilisticStampfli} and Thm.~\ref{thm:DoublyProbabilisticStampfli}) open possibilities beyond Tsirelson's problem itself.
It is natural to explore probabilistic commutation hypotheses in intrinsically infinite-dimensional settings, for instance, by sampling random unitaries using frameworks like free probability theory~\cite{mingo2017free}.
Viewing Stampfli's theorem as a generalization of Schur's lemma, another promising direction involves investigating whether our probabilistic versions can lead to analogous generalizations of other consequences of Schur's lemma, such as probabilistic Schur-Weyl duality, which could in turn lead to probabilistic formulations of quantum de Finetti theorems~\cite{renner2008security}.

\section*{Acknowledgments}
We thank an anonymous reader whose suggestions reduced the $d^2$ factor in Lemma~\ref{lem:ApproxSchur} to $d$ and inspired our use of Kemperman's theorem in developing new versions of probabilistic Stampfli's theorems.
The authors thank Scott McCullough for providing feedback on the manuscript.

X.X. and M.-O.R. acknowledge funding from the INRIA and the CIEDS in the Action Exploratoire project DEPARTURE.
X.X. and M.-O.R. acknowledge funding from the ANR through the JCJC grant LINKS (ANR-23-CE47-0003) and T-ERC QNET (ANR-24-ERCS-0008).
I.K. was supported by the Slovenian Research Agency program P1-0222 and grants J1-50002, N1-0217, J1-3004,
J1-50001, J1-60011, J1-60025.
Partially supported by the Fondation de l'École polytechnique as part of the Gaspard Monge Visiting Professor Program.
IK thanks École Polytechnique and Inria for hospitality during the preparation of this manuscript.
X.X., M.-O.R., and I.K. acknowledge funding from the European Union's Horizon 2020 Research and Innovation Programme under QuantERA Grant Agreement no. 731473 and 101017733 {\normalsize\euflag}.

\printbibliography

\end{document}